\documentclass[11pt]{article}
\usepackage[hmargin=1in,vmargin=1in]{geometry}
\PassOptionsToPackage{hyphens}{url}\usepackage[bookmarks,breaklinks=true,colorlinks=true, citecolor=black, urlcolor=black, linkcolor=black]{hyperref}
\usepackage{amsmath,geometry,bm,latexsym,amsfonts,graphicx,graphics,caption,setspace,subfig,longtable,morefloats}
\usepackage{rotating,color,natbib,soul}

\usepackage{mathptmx}
\usepackage[multiple]{footmisc}
\usepackage[normalem]{ulem}
\usepackage{lipsum}

\graphicspath{{figs/}}

\allowdisplaybreaks

\def\sym#1{\ifmmode^{#1}\else\(^{#1}\)\fi}


\renewcommand{\thetable}{\Roman{table}}

\usepackage[protrusion=true,expansion=true]{microtype}	
\usepackage{amsmath,amsfonts,amsthm} 

\usepackage{empheq,multirow,array,multicol}
\usepackage{gensymb}
\usepackage{comment,soul}

\newtheorem{theorem}{\normalfont \bfseries Theorem}

\newtheorem{corollary}{\normalfont \bfseries Corollary}

\newtheorem{remark}{\normalfont \bfseries Remark}

\newtheorem{lemma}{\normalfont \bfseries Lemma}

\begin{document}
\title{Renewable Power Trades and Network Congestion Externalities\thanks{First Draft: November 30, 2018. Nayara Aguiar: University of Notre Dame, Notre Dame, IN 46556. Email: ngomesde@nd.edu. Indraneel Chakraborty: University of Miami Business School, Coral Gables, FL 33124. Email: i.chakraborty@miami.edu. Vijay Gupta: University of Notre Dame, Notre Dame, IN 46556. Email: vgupta2@nd.edu.}}
\singlespace
\author{Nayara Aguiar \and Indraneel Chakraborty \and Vijay Gupta}

\maketitle
\onehalfspace
\begin{abstract}
\noindent Integrating renewable energy production into the electricity grid is an important policy goal to address climate change. However, such an integration faces economic and technological challenges. As power generation by renewable sources increases, power transmission patterns over the electric grid change. Due to physical laws, these new transmission patterns lead to non-intuitive grid congestion externalities. We derive the conditions under which negative network externalities due to power trades occur. Calibration using a stylized framework and data from Europe shows that each additional unit of power traded between northern and western Europe reduces  transmission  capacity for the southern and eastern regions by 27\% per unit traded. Such externalities suggest that new investments in the electric grid infrastructure cannot be made piecemeal. In our example, power infrastructure investment in northern and western Europe needs an accompanying investment in southern and eastern Europe as well. An economic challenge is regions facing externalities do not always have the financial ability to invest in infrastructure. Power transit fares can help finance power infrastructure investment in regions facing network congestion externalities. The resulting investment in the overall electricity grid facilitates integration of renewable energy production.  
\end{abstract}

\vspace{0.5in}



\thispagestyle{empty}
\clearpage
\setcounter{page}{1}
\doublespace

A higher dependence on renewable energy for electricity generation can help reduce greenhouse gas emissions. However, integrating renewable energy sources into the electricity grid comes with its own set of economic and technological challenges. Traditionally, electric power is mostly produced and consumed in relatively close proximity, as non-renewable power plants have been built close to consumption to minimize transmission power losses.\footnote{As an example, in the US, there are 66 balancing authorities that manage the balance of supply and demand in the regions they are responsible for. Such co-location means that only residual power is traded over the electric grid over long distances to ensure reliability of the power supply.} However, as renewable power becomes a significant source of energy, longer-distance transmission over the electric grid is becoming more frequent.\footnote{For example, areas with more solar and wind power in the center of the U.S. need to be connected to U.S. coastline load centers. Power from the wind-rich west Texas needs to be supplied to the load centers in east Texas.} Balancing supply and demand of power in  presence of inherent variability in  renewable energy generation has also led to rapid increases in inter-balancing authority power transfers that occur over longer distances. Increasing power transmission over regions that had not traded much until now requires investments in the electric grid infrastructure.\footnote{For example, new transmission lines have been constructed to connect west and east Texas.} 

However, the determining the financing mechanism for such an investment in power infrastructure is an economic challenge. Questions include who should pay for investments in one part of the grid, and how should the network externalities caused by long distance power trades be priced? In contrast to other types of trade networks, power trade over electrical networks face unique non-intuitive network externalities driven by physical laws (Kirchoff's laws) that can complicate both the placement and the cost sharing of such investment decisions \citep{griffin2009electricity}. Over an electrical grid, power trade between two regions may be constrained or facilitated due to trade between two other regions. As an example, North-North trade of power can potentially reduce the trading capacity of South-South trade of power. 

The economic policy implication of the network externalities discussed in this paper is that investment in electrical grids cannot be made piecemeal where only the power generation facilities and power consumption areas are connected with better transmission lines. Energy policymakers could potentially finance the investment in electric grid infrastructure in regions that face externalities but not the benefits using a transit fare that matches the magnitude of network externalities imposed.

Our paper makes two contributions. First, we derive analytical expressions to identify cases when power trade externalities bind on other electrical grid regions. To this end, we decompose a power flow due to trades between two nodes into flows that would exist in a network that trades any other goods and the additional trade component due to Kirchoff's laws. This decomposition allows us to estimate the magnitude of the electric grid externalities due to Kirchoff's laws as a fraction of the total power trade, which is effectively a measure of the efficiency of a power trade. We then analyze various trade scenarios over the network for feasibility, as  network externalities due to one power trade interacts with other power  trades over the network due to physical laws.  The feasibility region analysis delivers the conditions under which a power trade reduces the network capacity for other trades. The analysis also provides conditions under which, perhaps unintuitively, a power trade increases the network capacity for other power trades. 


As the second contribution, we calibrate our analytical framework to estimate the magnitude of  externality-driven congestion constraints on the  grid. Our calibration uses data from the European power grid since longer distance power transmission over the grid is already frequent in the continent. European grid participants have expressed  concerns regarding related network externalities.\footnote{See ``In Central Europe, Germany's Renewable Revolution Causes Friction," \textit{The Wall Street Journal}, Feb 16, 2017.  \url{https://www.wsj.com/articles/in-central-europe-germanys-renewable-revolution-causes-friction-1487241180}.} To bring the data close to the analytical framework, we divide the European grid into four regions.\footnote{The northern region includes Germany and Scandinavia in our exercise. Eastern Europe includes Central European countries such as Poland and Hungary. Southern Europe includes includes Italy and Greece. Western Europe includes France, Spain, and Portugal.} Our power transmission data suggests that the European grid obtains surplus power from northern  and eastern Europe. Southern  and western Europe consume this surplus. The  resulting long distance transmission of power between northern and western Europe has significant negative network externalities for the European grid. Our calibrated framework  estimates that the marginal power trade between the northern and western portions of the grid  reduces grid capacity by approximately 27\% for southern and eastern Europe. The estimates are for a typical period in Spring 2019, when the grid system did not face significant strains due to heating or cooling related spikes in consumption. 

The estimated results are for the grid at present. The increasing importance of renewable energy as a source of power production implies that  the average distance between power production facilities and power consumption centers will keep  increasing over the next few decades. This, in turn, implies that grid congestion due to physical laws will also increase over time. Hence,  our results suggest to policymakers that as more renewable sources of energy production join the electric grid, to ameliorate potential negative externality of longer distance power trades, new investments in the electricity grid infrastructure need to be made.  Further, these investments are required across the grid, not just between generation facilities and consumption areas given the network externalities of power trades.   




We start our analysis with a stylized framework where two power producers and consumers are connected in a cyclic network. Without loss of generality, the producers and consumers can be thought of as net producers and consumers. In a network that trades goods, trade between two connected parties only causes congestion on the link between the two parties. There are no externalities in terms of congestion on other parties. 

However, in an electrical grid, Kirchoff's laws dictate that power flows in a cyclic manner, i.e. even if the North producer wishes to send power just to the North consumer, the electric circuit must close and power flows through the cyclic network, i.e. along the Southern link as well. Thus, depending on the direction of power flow, trade between two parties may create or ameliorate congestion on links between other producers and consumers, in this case the Southern producer and consumer. The stylized framework allows us to decompose the network power flow into two additive and separable terms: network flow if power were any other good and the externality imposed on the network if the traded good is electricity. 

Next, we characterize the set of feasible trades in our framework. Investigating  feasibility of trades allows us to evaluate the impact of power trade externalities on trade constraints. For a trade to be feasible, the network flows generated  must satisfy the capacity constraints of all the network edges. We derive relations between line capacities and reactances which restrict the tradeable set. 

To quantify the network externalities, we apply the analytical results obtained on data from the European Power Grid. Along with economic integration, the electricity markets of Europe have also integrated. The benefits of electricity market integration are similar to those of other market integrations \cite{Krugman:79,Krugman:80}.  However, our paper shows that trade between two nodes can create transmission constraints on other lines in the grid. The northern European region has made great strides in producing renewable energy and this energy has created large power trades over the electrical grid. Thus, Europe provides an ideal setup for us to quantify the magnitude of power trade externalities on the electrical grid. The setup also allows us to determine the threshold where trades become infeasible over the grid if additional power flows were to take place.

We find that the ratio of trades on the North-North (Germany-France) line is approximately 62\% more than the South-South (Eastern Europe-Italy) line. Given the effective reactances of the grid, we first establish that the North trade restricts the South trade. We then decompose the flows into goods flow and additional component due to physical laws that exist in power flows. We find that the present trades reduce the transmission capacity of the South-South line by an amount equal to 10.3\% of the North-North trade. {The marginal trade between the Northern nodes reduces capacity for the southern nodes by approximately 26.9\% per unit traded.}

Regarding feasibility, we find that the trade feasible set shrinks by 4\% due to physical laws compared to the case where power flows as any other commodity. Note that this calculation considers all feasible trades,  unadjusted for the probability of respective trades. In the setup calibrated for European data, we find that power trades face 8.7\% lower grid capacity that goods trades. This is an economically significant effect of physical laws. These results suggest that newer investment take the network externalities into account and expand grid transmission capacity across the network to avoid grid congestion.


Seminal work in energy economics addresses the question of appropriate price policy in regulated industries. Specifically, how we can make sure that the price leads to the correct amount of physical capacity of production and efficient utilization to cover the full social costs of the resources used \citep{Steiner:57,Williamson:66}. \cite{Littlechild:70,CrewKlei:71} consider the question of joint product pricing when multiple energy producers with different cost characteristics work together to meet demand. The aforementioned classic papers consider the problem of joint generation and pricing, but not of trades. Our paper contributes by pointing out network externalities of power trades and the marginal cost of such trades which can facilitate pricing of grid resources. 

Significant changes have taken place in the electricity generation and transmission industry since the 1990s \citep{DaviWolf:12,BoreBush:15}. Regulatory changes have allowed for a marketplace for trading power compared to vertically integrated monopolies in the decades before. Generation of power has also become more distributed as non-utilities have started selling power. As renewable energy becomes more feasible as an alternative source of power, these changes that started in the 1990s have only sped up. Yet, challenges remain in the integration of renewable energy into the grid. \cite{Edenetal:13} investigate the challenges in integrating renewable energy due to inherent variability in wind and solar power production. \cite{Koketal:20} consider the joint problem of investment in renewable and conventional energy production facilities by the same firm. \cite{QiWeetal:15,Cruietal:19}  analyze the benefits of including storage along with renewable power to attain desired production levels. Adoption of cleaner non-renewable energy such as natural gas and nuclear energy also faces challenges \citep{DaviMueh:10,Davis:12}. 

In this context, power trades over the grid provide an important method of integrating new energy sources. However, in this case, sharing the costs of investment in renewable energy production and  grid improvements remains a difficult problem. Regarding renewable energy production, the level of subsidies provided for renewable energy production need to be carefully determined. \cite{Cullen:13} shows that the value of emissions offset by wind power exceeds cost of renewable energy subsidies only if social costs of pollution are very high. \cite{Eglietal:18} suggest that costs of renewable energy production are not declining sufficiently to make them feasible without subsidies. Regarding transmission, \cite{DeCaDutr:13} discuss the challenges in the adoption of smart grids: investments required are large and the benefits do not accrue to potential investors directly. Our paper provides a path forward for integrating renewable energy into the grid: by pricing the externalities of long-distance power trades. The transit fare that we suggest can help foster investments to improve the grid. Given that it is hard to target taxes or subsidies efficiently in the energy markets \citep{Allcott:15}, our paper contributes to reducing externality-driven distortions in this market.  


Section~\ref{sec:background} provides a brief background of power flow in electric grids. Section~\ref{sec:setup} sets up the analytical framework. Section~\ref{sec:feas_regions} characterizes the feasible region that limits the allowable trades in the power and goods networks. Section~\ref{sec:eurogrid} applies the analysis in previous sections, utilizes data from Europe's electric grid to calibrate the setup, and estimates the network congestion externalities in the European grid. Section~\ref{sec:conc} concludes. 

\section{Electric Grids: Past, Present and Future}
\label{sec:background}


The first electric utilities were developed in isolation of one another, granting them de facto monopolies. These electric grids became interconnected by the end of the 1920s, and the possibility of sharing power to meet demand across larger areas became a reality \citep{borbely}. Even though electric grids were interconnected, the traditional model for electricity markets was a vertically integrated monopoly until the 1990s. This implies that a single entity owned the entire chain of services needed to provide electricity to customers -- generation, transmission, distribution, and retail. Local utilities were responsible to set retail tariffs, and hence these markets were \emph{regulated}. However, these regulatory approaches were argued to be inefficient, leading to slow rates of innovation and increased energy price for consumers. In 1994, California led the  reforms for the deregulation of electricity supply in the United States \citep{warwick}. In this \emph{deregulated} model, generators compete against each other and markets are created to trade power \citep{ilic}.

The interconnection of electric grids brings reliability benefits that allow for the delivery of on-demand electricity. In the aftermath of World War II, the need to efficiently use limited energy resources led to an increase in the interconnection of the European national grids. This new operational paradigm allowed countries to purchase power from each other in case of energy shortages. This paper utilizes the European grid as a case study to analyze the implications of the interconnection of electric grids as the supply and demand patterns shift due to increasing renewable penetration.    

When grid infrastructures are highly interconnected, as is the case in the European grid, large amounts of power can be transported from where it can be produced cheaply to centers of high demand. Achieving this, however, is not only a technical but also an institutional challenge, as markets need to adapted for coordinated operations across borders \citep{verzijlbergh}.

With the greater integration of renewable energy sources into the grid, the power systems of the present and future face new challenges. There are several aspects involved in the cost of integrating more renewable generation in the electricity mix, including the costs to connect new wind and solar power plants to the grid and potentially expanding transmission and distribution lines; the costs to keep enough generation capacity to balance renewable variations; and the impact new renewable generation brings to already existing power plants \citep{agora}. 

\section{Setup}
\label{sec:setup}

We consider a stylized setup with four nodes that form a cyclic network. In this setup, we analyze the flows that result from trades in two different types of network: an electricity grid and a goods trade network. Our approach helps us understand the potential negative externalities imposed when parties trade power compared to other tradeable goods. 

Analyses of power systems as economic systems use Direct Current (DC) power flow models as opposed to the less tractable Alternating Current models \citep{purcetal:05}.  We also use a DC power flow model. DC power flows across all transmission lines in a network can be found using a linear transformation of the vector of power injections in each node. 


We begin by presenting the optimal power flow problem in the electricity network in Section~\ref{sec:opf}. Section~\ref{sec:flows} contrasts the power flow solution with a goods flow solution. The section shows that the network power flow solution can be written as an additive and  separable function of goods flow solution and externalities in an electric network. 

\subsection{The Power Flow Problem}\label{sec:opf}

Consider the network in Fig.~\ref{fig:4node}. We assume that each node represents a region $i \in \{NW,NE,SE,SW\}$\footnote{The regions are denoted northwest (NW), northeast (NE), southeast (SE), and southwest (SW).} that produces and consumes power with a potential deficit or surplus. The surplus power can be traded across the network. Trade links $\ell \in \{N,E,S,W\}$\footnote{The links are denoted north (N), east (E), south (S), and west (W).} connect nodes and have constrained transmission capacity $C_{\ell}$, limiting trades. 

To illustrate the externalities imposed on the network due to power trades, we need two sets of trades. Thus, the stylized model assumes that two nodes $NW$ and $SE$ produce surplus power, which is consumed by the remaining two nodes $NE$ and $SW$, respectively.\footnote{A negative surplus power means a net consumption. Thus, the North-North trade can also have node $NW$  as a net consumer and node $NE$ as a net producer. The same reasoning applies to the South-South trade.} Thus, there is a North-North ($NE$ and $NW$) and a South-South ($SW$ and $SE$) power trade over the grid in the figure.


\begin{figure}[h]
    \centering
    \includegraphics[width=0.35\textwidth]{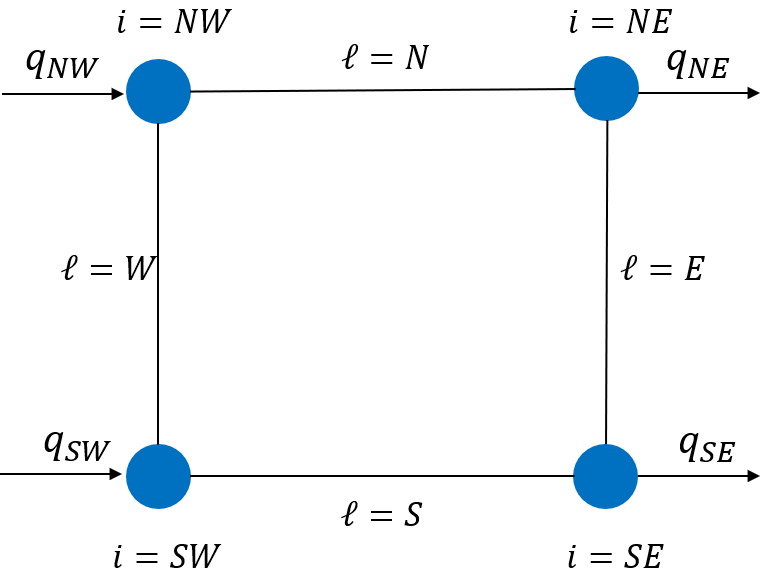}
\caption{\small Four-node network setup. Each node $i \in \{NW,NE,SE,SW\}$ is connected to two other nodes through an edge $\ell \in \{N,E,S,W\}$ with constrained transmission capacity $C_\ell$. In the electricity network, each node $i$ has a net power injection $q_i$ which is positive (resp. negative) if the node is supplying power to (resp. demanding power from) the network.}
\label{fig:4node}
\end{figure}



The electricity grid is designed such that the amount of power traded over the network is maximized, conditional on the surplus and deficit of nodes:
\begin{align}
    \underset{\mathbf{q}}{\max} \,\  \sum \mathbf{|q|} & \nonumber  \\
     \mathbf{1^T q} &= 0 \nonumber \\
    \label{eq:gridobj} -\mathbf{C} \leq \mathbf{Sq} &\leq \mathbf{C},
\end{align}
where the first constraint ensures that the sum of surplus power injected into the grid equals the sum of power deficit of consumer nodes. The second constraint is  a budget constraint to ensure that the total power transmitted through the network is within network capacity $\mathbf{C}$. As power can flow in either direction across the transmission lines, the inequality is on the element-wise magnitude of the flows.

The transformation matrix $\mathbf{S}$  is known as the shift-factor matrix \citep[see, for example][]{Sahretal:17}. The matrix helps determine the power flows in the network due to the power injection vector $\mathbf{q}$ so that the flow dynamics obey  Kirchoff's laws. Appendix~\ref{app:Smatrix} provides a detailed explanation of how to compute the shift-factor matrix. 



\subsection{Power Trades and Goods Trades}\label{sec:flows}

To investigate the objective of the electric grid compared to a goods network, we need to estimate the effect of network externalities on power trade. To this end, this section breaks down a power trade into the flows that would exist in a network that trades any other good, and the additional externality-driven network flow if the traded item is electricity. Such a decomposition of the transformation matrix allows us to understand the relative magnitude of externality in practice. 

As discussed above, we consider the North-North, South-South power trade case. Thus, the net injection vector is $\mathbf{q} = (q_N,-q_N,q_S,-q_S)$. In other words, nodes $NE$ and $NW$ trade amount $q_N$, and nodes $SW$ and $SE$ trade $q_S$ units of power\footnote{In other words, $q_{NW}=-q_{NE}=q_N$ and $q_{SE}=-q_{SW}=q_S$.}. The network topology is cyclic, as shown in Figure~\ref{fig:4node}. By convention, we consider flows through each network link to be positive if they go from node $NW$ to $NE$, from $NE$ to $SE$, from $SE$ to $SW$, and from $NW$ to $SW$. Flows are denoted negative if they flow in the opposite directions.

\textbf{Goods trade case:} In a network that trades goods, the resulting flows $f^g_{\ell}$ through each link $\ell$ are as follows:
\begin{equation}\label{eq:fpairwise}
\bm{f^g} =
\begin{bmatrix}
    f^g_{N} \\ f^g_{E} \\ f^g_{S} \\ f^g_{W}
\end{bmatrix} =
\begin{bmatrix}
    q_N \\ 0 \\ q_S \\ 0
\end{bmatrix} =
\begin{bmatrix}
   1 & 0 & 0 & 0\\
   0 & 0 & 0 & 0\\
   0 & 0 & 1 & 0\\
   0 & 0 & 0 & 0
\end{bmatrix}
\begin{bmatrix}
    q_N \\ -q_N \\ q_S \\ -q_S
\end{bmatrix}.
\end{equation}
For trades that generate flows $\bm{f^g}$ to be feasible, the capacity constraint of each link in the network needs to be satisfied, that is  $-\mathbf{C} \leq \bm{f^g} \leq \mathbf{C}$. 


\textbf{Power trade case:} In case of an electricity grid, the transformation matrix $\mathbf{S}$ that maps power  injections in each node to the flows in each line is based on Kirchoff's Laws. Let $x_{\ell}$ represent the reactance of the transmission line $\ell$.\footnote{
The power flows are inversely proportional to the reactance of the transmission line. A line with lower reactance can transmit more electric charge at a given voltage. Appendix~\ref{app:Smatrix} provides the relationship between power flows and shift-matrix $\textbf{S}$ in general.} The following lemma  provides the relationship between power flows over a  network and flows when the traded good is not power.



\begin{lemma}\label{lem:powerflows} \textbf{(Flow Decomposition)}
The power flows $\bm{f^p}=\mathbf{Sq}$ are such that
\begin{equation}\label{eq:flowrelation}
    \bm{f^p}=\bm{f^g} + \mathbf{S_\sigma}~\mathbf{q},
\end{equation}
where $\bm{f^g}$ are the flows that are generated for non-electric trades. The second term is a vector of additional flows that arise in the case of power trades, and is given by
\begin{equation}
    \mathbf{S_\sigma}~\mathbf{q} = \frac{1}{\sum x_{\ell}}\begin{bmatrix}
    -(x_{N}+x_{E}+x_{S}) & -(x_{E}+x_{S}) & -x_{S} & 0\\
    x_{W} & x_{N}+x_{W} & -x_{S} & 0\\
    x_{W} & x_{N}+x_{W} & -x_{S} & 0\\
    x_{N}+x_{E}+x_{S} & x_{E}+x_{S} & x_{S} & 0
\end{bmatrix}
\begin{bmatrix}
    q_N \\ -q_N \\ q_S \\ -q_S
\end{bmatrix}
\end{equation}
\end{lemma}
\noindent Proof is in Appendix~\ref{app:proof_lem1}.

As is the case for trade in goods, the absolute value of power flows over each link in the network is bounded by the respective transmission line capacity $\textbf{C}$. For feasible goods flows, the relationship \eqref{eq:flowrelation} can be used to bound power flows as follows:
\begin{align}
    \left|~\bm{f^p}~\right| &= \left|~\bm{f^g}+ \mathbf{S_\sigma}~\mathbf{q}~\right|  \leq \left|~\bm{f^g}~\right|+  \left|~\mathbf{S_\sigma}~\mathbf{q}~\right|     \label{eq:uboundpf} \leq \mathbf{C}+  \left|~\mathbf{S_\sigma}~\mathbf{q}~\right|,
\end{align}
The inequality above considers that the net injection vector $\mathbf{q}$ and the network capacity $\textbf{C}$ are the same for both goods trade and electricity grid cases. The upper bound shows that  power flows in an electricity grid can rise \textit{above} the capacity limitations imposed by individual links of vector $\textbf{C}$, which will bind for goods trade. This is because of the second term $\mathbf{S_\sigma q}$ that is an externality determined by the power trades between all the nodes in the network. As we discussed before and discuss further later, the externality may be positive or negative depending on network parameters and direction of power flow.


\subsection{Network externality of power flows} Let these additional flows $\mathbf{S_\sigma q}$, that exist in the electricity grid  but not in the goods trade network, be denoted as $\sigma$-flows. 
For the same amount of trade, higher capacity of the electric network over a goods network is not the only possible outcome due to power network externalities. The power flows in the electricity grid can also make the capacity bounds in the transmission lines \textit{tighter} than in the goods trade case. This situation happens when the externalities-driven $\sigma$-flows are in the same direction as the flows from the intended trades between producer and consumer nodes. 

Theorem~\ref{th:sigmaf} characterizes the magnitude and direction of $\sigma$-flows in the network as well as the conditions under which the electric grid has lower trade capacity than a comparable goods network.
\begin{theorem} \label{th:sigmaf}
\textbf{($\sigma$-Flows Properties)}
Consider a network topology as in Figure~\ref{fig:4node} and a net injection vector of the form $\mathbf{q} = (q_N,-q_N,q_S,-q_S)$. Let $f^{\sigma }_\ell$ be the $\sigma$-flow along the transmission line $\ell$.

(a) The $\sigma$-flows in each transmission line are equal in magnitude and given by
\begin{equation}\label{eq:fsigma}
    |f^{\sigma }_\ell|=\frac{|~x_{N}q_N+x_{S}q_S~|}{\sum x_\ell}~~\forall \ell \in \mathcal{L}.
\end{equation}

(b) The direction of the $\sigma$-flows in the network is clockwise if $~(x_{N}q_N+x_{S}q_S) < 0$, and counterclockwise if $~(x_{N}q_N+x_{S}q_S) > 0$\footnote{In our model, a clockwise flow direction corresponds to the nodal path $NW-NE-SE-SW-NW$. The reverse path is referred as counterclockwise.}.

(c) If node $NW$ is a consumer (supplier) when node $SE$ is a supplier (consumer), the $\sigma$-flows will lead to a tighter bound in transmission line North if $|\frac{q_N}{q_S}|<\frac{x_{S}}{x_{N}}$, and in transmission line South if $|\frac{q_N}{q_S}|>\frac{x_{S}}{x_{N}}$.
\end{theorem}
\noindent Proof is in Appendix~\ref{app:proof_th1}.

The setup is of a North-North and a South-South trade. Part (a) calculates the magnitude of the externality flow in the electric network. The magnitude discussed in part (a) and direction of the flow discussed in part (b) depends on the relative direction  of the two trades. If node $NW$ consumes and node $SE$ produces or vice versa, i.e. the North-North and the South-South trade are both in the same direction, then the sigma flows cause a negative externality and reduce network capacity. 

Part (c) provides the condition under which the externality is imposed on the North-North trade or the South-South trade. The condition is effectively the relative magnitude of trades ($|x_Nq_N|<>|x_Sq_S|$). The theorem yields a corollary which provides a razor-edge case when the power trade network will not create an externality. 

\begin{corollary} \label{col:floweq}
\textbf{(Trade Flow Equivalence)}
The electricity and the goods trade grids generate the same flows if and only if node $NW$ is a consumer (supplier) when node $SE$ is a supplier (consumer), and $q_S=-\frac{x_{N}}{x_{S}}q_N$. In the degenerate case with  $\mathbf{q} = \mathbf{0}$, there are no trades in the network, and thus the flows are identically zero.
\end{corollary}

\begin{proof}
From \eqref{eq:flowrelation}, we note the flows will be equivalent for both cases if and only if $\mathbf{S_\sigma q} = \mathbf{0}$. Using \eqref{eq:fsigma}, we find that all the $\sigma$-flows will be zero if $q_S=-\frac{x_{N}}{x_{S}}q_N$. When the net injections are zero, we have a degenerate case in which $\bm{f^p}=\bm{f^g}=\mathbf{0}$.
\end{proof}

We note that if the conditions for flow equivalence established in Corollary~\ref{col:floweq} do not hold, then the vertical lines of our network ($\ell=E$ and $\ell=W$) will have a tighter bound in the electricity grid case due to non-zero $\sigma$-flows flowing through them. This situation is illustrated in Fig.~\ref{fig:Sflows}. 

\begin{figure}[!hbt]
  \centering
    \includegraphics[width=.8\textwidth]{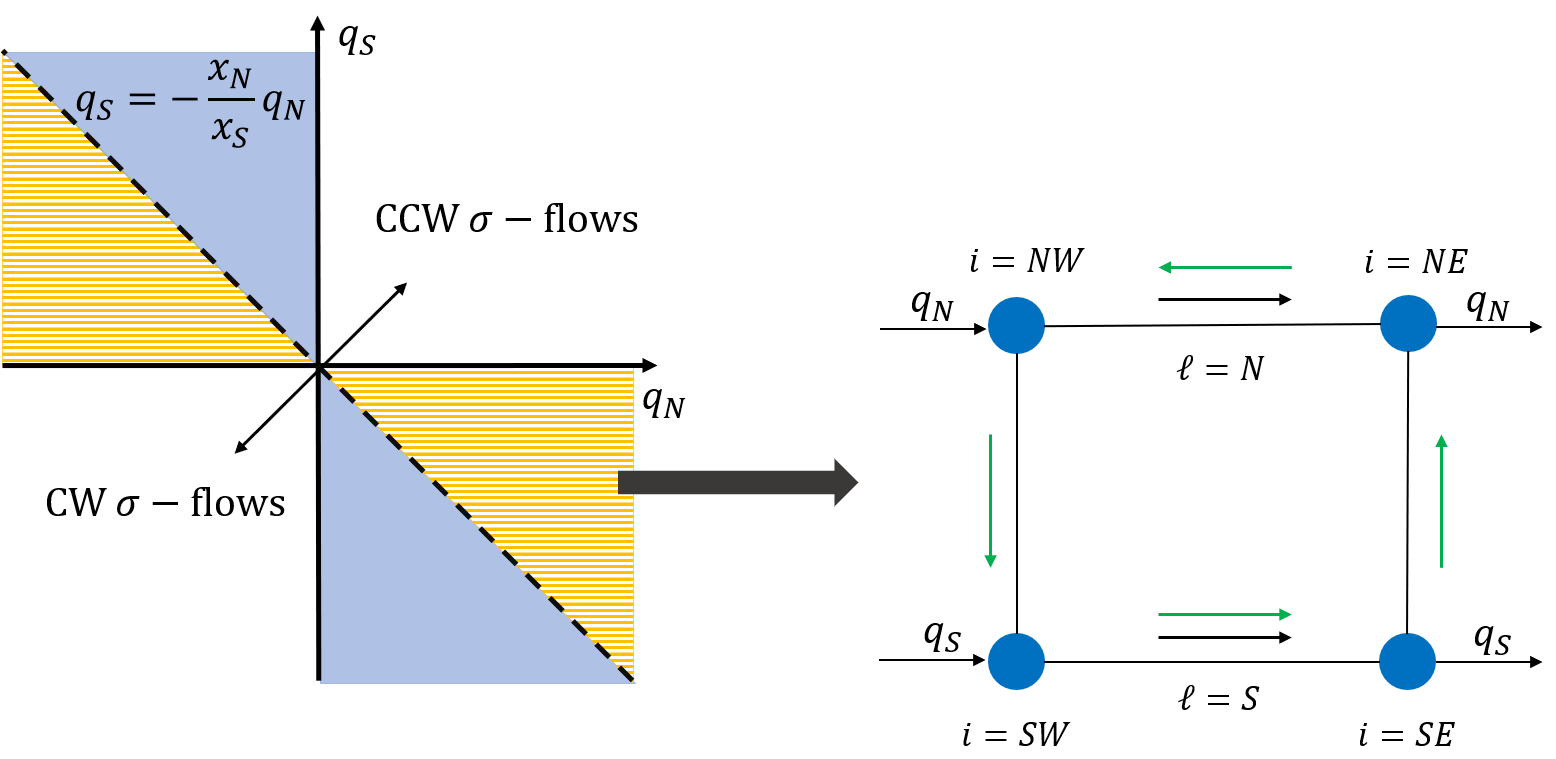}
    \caption{\small The flows that arise  due to power trades in the four-node network shown in Fig.~\ref{fig:4node} can be decomposed into two terms: the first is equivalent to the flows that would have been generated if the trades were pairwise, and the second term comprises extra flows of equal magnitude through each line, which we call $\sigma$-flows. Let the trades in this network be of the form $\mathbf{q} = (q_N,-q_N,q_S,-q_S)$. The dashed diagonal line in the ($q_N,q_S$) plane plot is given by $x_{N}q_N+x_{S}q_S = 0$. As shown in Corollary~\ref{col:floweq}, the flows in the power and goods networks are equivalent for all the trades on this line. The $\sigma$-flows are clockwise (CW) if $(x_{N}q_N+x_{S}q_S) < 0$ and counterclockwise (CCW) otherwise. When nodes $i=NW$ and $i=SE$ take different roles, that is, one is a supplier when the other one is a consumer, the $\sigma$-flows will cause transmission line $\ell=N$ to have a tighter bound (as compared to a goods network) if $|\frac{q_N}{q_S}|<\frac{x_{S}}{x_{N}}$, while transmission line $\ell=S$ will be the one with a tighter bound if $|\frac{q_N}{q_S}|>\frac{x_{S}}{x_{N}}$. In this figure, these regions are mapped in the ($q_N,q_S$) space, where the yellow striped region corresponds to transmission line $\ell=N$ being tighter, while the blue region refers to when line $\ell=S$ has a tighter bound. The network illustration exemplifies the decomposition of flows for a pair of trades ($q_N,q_S$) in the yellow region in the fourth quadrant. The dark arrows in lines $\ell=N$ and $\ell=S$ refer to the component equal to the flows that would be generated if non-power goods were traded, and the green arrows are the CCW $\sigma$-flows. We can observe that the externalities introduced by the $\sigma$-flows are such that the transmission line $\ell=N$ expands its capacity (extra flow opposes goods flow), while line $\ell=S$ has a tighter capacity (extra flow reinforces goods flow).}
    \label{fig:Sflows}
\end{figure}

These externalities  impose trade constraints in certain cases. The next section  characterizes the feasible sets of trades over the four-node network that we consider. 

\section{Power trade feasibility over a network}
\label{sec:feas_regions}


For a trade to be feasible, the flows it generates over the network must satisfy the capacity constraints of all the network edges. This condition is included as a budget constraint in the trade maximization problem expressed as equation system \eqref{eq:gridobj}. Investigating trade feasibility allows us to evaluate the impact of power trade externality on trade constraints.  The feasible sets we will be analyzing in this problem are illustrated in Figure~\ref{fig:feasiblesets}.


\subsection{Feasible trades per line}
\label{sec:linefeasibility}

The net trades are as before $\mathbf{q} = (q_N,-q_N,q_S,-q_S)$.  Writing out the flows expressions \eqref{eq:fpairwise} and \eqref{eq:flowrelation} as a function of the traded quantities $q_N$ and $q_S$, we can define these feasible sets as follows. 
The feasible trades in a goods network are constrained only by the line capacity $C_N,C_S$: 
    \begin{equation}\label{eq:fp}
        \mathcal{F}_G = \left\{(q_N,q_S) \mid~|q_N| \leq C_N, |q_S| \leq C_S\right\}.
    \end{equation}

In the case of an electric network, the line capacities have to be sufficient to not only allow the trade between the nodes connected by the line, but also the additional power flows through the line due to trade between other nodes in the network driven by physical laws. For example, for line  $\ell=N$ the feasible trades set  $\mathcal{F}_{P1}$ is determined by a function of  trade  $q_N$ on the line, but also by a function of a trade $q_S$ on the South-South line: 
\begin{equation}\label{eq:f12}
            \mathcal{F}_{P1} = \left\{ (q_N,q_S)~\Big|~ \left|f^p_N\right| \leq C_N \right\} = \left\{ (q_N,q_S)~\Big|~ \left|\left(1-\frac{x_{N}}{\sum x_{\ell}}\right) q_N-\frac{x_{S}}{\sum x_{\ell}}q_S\right| \leq C_N \right\}
        \end{equation}
where $x_{\ell}$ represents the reactance of the transmission line $\ell$. Thus, the feasible trades set over the electric network is $\mathcal{F}_P = \bigcap_{k=1}^4\mathcal{F}_{Pk}$, where the remaining three lines are constrained by the following relations:
        \begin{eqnarray}\label{eq:f23}
            \mathcal{F}_{P2} &=& \left\{ (q_N,q_S)~\Big|~ \left|f^p_E\right| \leq C_E \right\} = \left\{ (q_N,q_S)~\Big|~ \left| -\frac{x_{N}}{\sum x_{\ell}}q_N-\frac{x_{S}}{\sum x_{\ell}}q_S\right| \leq C_E \right\} \nonumber \\
            \mathcal{F}_{P3} &=& \left\{ (q_N,q_S)~\Big|~ \left|f^p_S\right| \leq C_S \right\} = \left\{ (q_N,q_S)~\Big|~ \left| -\frac{x_{N}}{\sum x_{\ell}}q_N+\left(1-\frac{x_{S}}{\sum x_{\ell}}\right)q_S\right| \leq C_S \right\} \nonumber \\
            \mathcal{F}_{P4} &=& \left\{ (q_N,q_S)~\Big|~ \left|f^p_W\right| \leq C_W \right\} =  \left\{ (q_N,q_S)~\Big|~ \left| \frac{x_{N}}{\sum x_{\ell}}q_N+\frac{x_{S}}{\sum x_{\ell}}q_S\right| \leq C_W \right\}
        \end{eqnarray}

Figure~\ref{fig:feasiblesets} illustrates the interaction between power flows as they affect  feasible trades. If trades do not affect each other, as is the case in a goods network, the feasible region $\mathcal{F}_G$ is a rectangular 
area, as shown in panel (a). Thus, the magnitude of trade along the horizontal axis $q_N$ (North-North) does not affect the magnitude of trade along the vertical axis (South-South). The only constraint in goods trade is the capacity of the respective North-North and South-South links.  

\begin{figure}[t]
  \subfloat[]{
	\begin{minipage}[c][1\width]{
	   0.3\textwidth}
	   \centering
	   \includegraphics[width=1\textwidth]{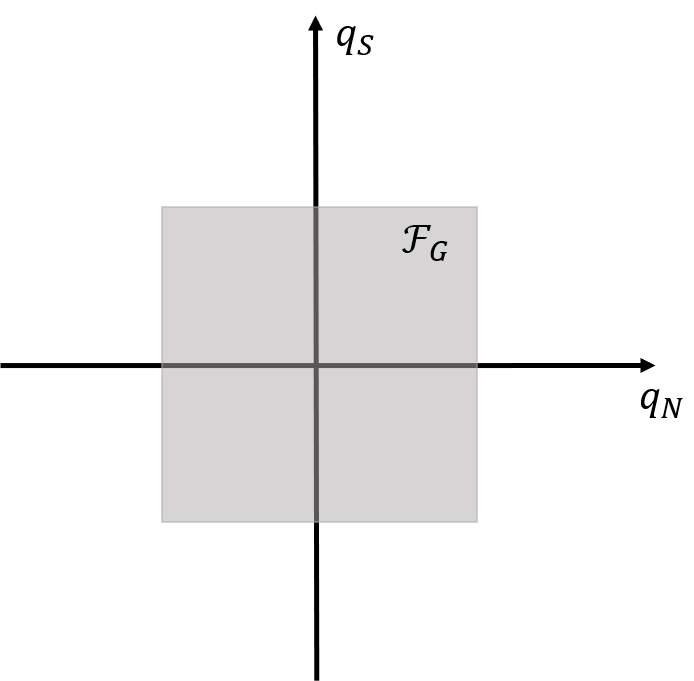}
	   \label{fig:Fg}
	\end{minipage}}
 \hfill 	
  \subfloat[]{
	\begin{minipage}[c][1\width]{
	   0.3\textwidth}
	   \centering
	   \includegraphics[width=1\textwidth]{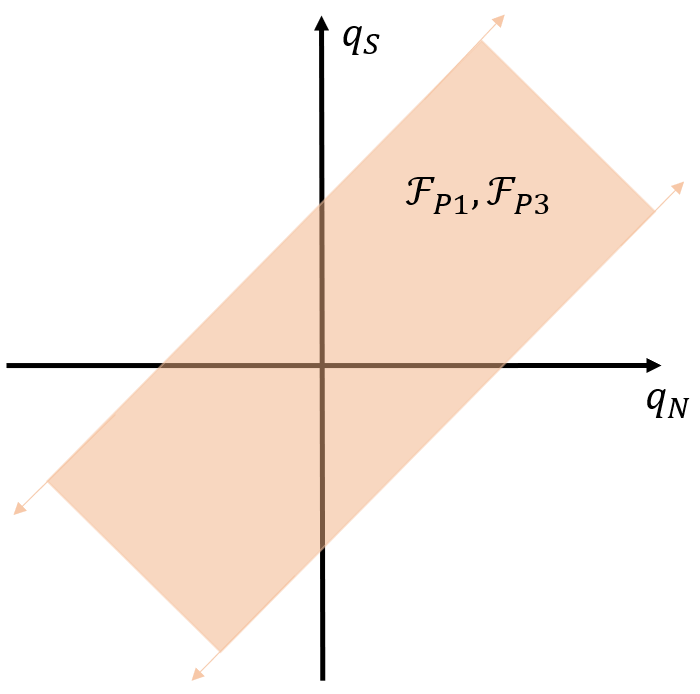}
	   \label{fig:Fp13}
	\end{minipage}}
 \hfill	
  \subfloat[]{
	\begin{minipage}[c][1\width]{
	   0.3\textwidth}
	   \centering
	   \includegraphics[width=1\textwidth]{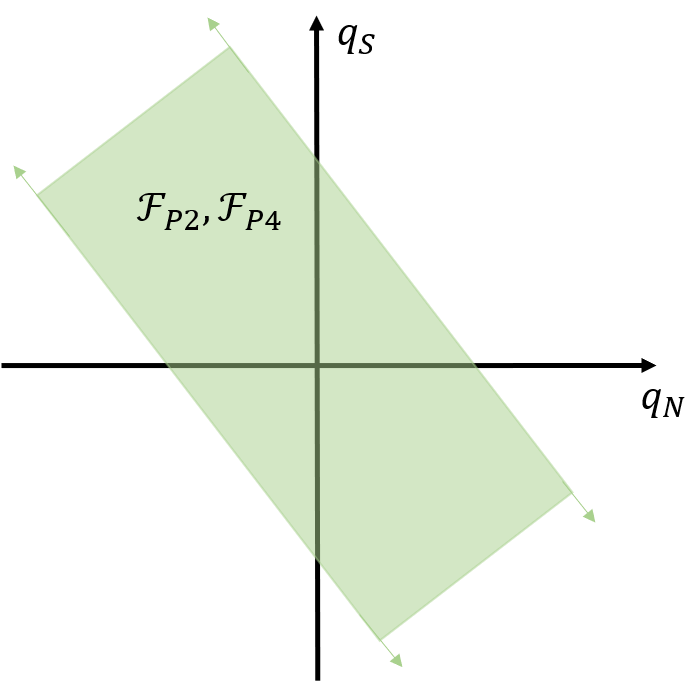}
	   \label{fig:Fp24}
	\end{minipage}}
\caption{Illustration of feasible regions for trades across the four-node network in Fig.~\ref{fig:4node}. For a trade to be feasible, the resulting flows must satisfy the capacity constraints of all lines in the network. If the trades performed are pairwise, the feasible region $\mathcal{F}_G$ is a rectangular
area, as shown in (a), which indicates independence between the trades. For example, trades between nodes $i=NW$ and $i=NE$ are constrained only by the capacity of the link which connects these nodes and not by the trades performed between the remaining two nodes. In other words, the flows that arise due to a certain trade do not cause externalities on the nodes that are not trading. If power trades are considered, the feasible region is modified due to externalities that each trade causes on the feasibility of further trades. As illustrated in (b) and (c), the feasible region for each transmission line $\ell$ in the electricity grid is given by the area between two parallel lines. For the horizontal transmission lines, the boundaries of this region have a positive slope, as in (b), while the feasible region for the vertical lines are bounded by parallel lines with a negative slope (c). The overall feasible region for power trades in the network studied $\mathcal{F}_P$ is the intersection of the feasible regions for each individual line. This region will be characterized in detail in Section~\ref{sec:feas_regions}, where we will show that it can either be a parallelogram or a hexagon.}
\label{fig:feasiblesets}
\end{figure}

However, such trade independence is not the case with power trades. Panel (b) in Figure~\ref{fig:feasiblesets} illustrates the implications of equations~\eqref{eq:f12} and \eqref{eq:f23}. The positive slope for $\mathcal{F}_{P1},\mathcal{F}_{P3}$ suggests that there are positive complimentarities in the North-North and South-South trades $q_N,q_S$. In other words, as $q_N$ increases, the two lines support a larger trade $q_S$ in the same direction as well.  From the line equations above, we find that the feasible sets  $\mathcal{F}_{P1},\mathcal{F}_{P3}$   have a positive slope $m_1,m_3$ of:
\begin{eqnarray}\label{eq:m1}
   m_1&=&\frac{x_{E}+x_{S}+x_{W}}{x_{S}}>1 \nonumber \\
     m_3&=&\frac{x_{N}}{x_{N}+x_{E}+x_{W}}<1.
\end{eqnarray}
Thus, the feasibility condition is driven by lines reactances $x_i$. However, while the North-North and South-South lines have positive complimentarities, this is not the case for the two North-South lines. $\mathcal{F}_{P2}$ and $\mathcal{F}_{P4}$ are limited by lines with negative slope:
\begin{equation}\label{eq:m24}
   m_2=m_4=-\frac{x_{N}}{x_{S}},
\end{equation}
Panel (c) in Figure~\ref{fig:feasiblesets} illustrates this point. In other words, the sets  $\mathcal{F}_{P2}$ and $\mathcal{F}_{P4}$ are least restricted when trades have opposite signs, i.e. when nodes $NW$ and $SW$ are net producers and nodes $NE$ and $SE$ are net consumers or vice-versa. This in intuitive as flows have to complete the cycle.\footnote{Note that the illustrations in Figure~\ref{fig:feasiblesets} only show the shape of each feasible region, and the actual slopes and intersections of the boundaries with the axis will vary with the network parameters. Thus, although $\mathcal{F}_{P1}$ and $\mathcal{F}_{P3}$ have the same shape (area within two parallel lines with positive slope), they will not necessarily be equivalent. The same reasoning goes for $\mathcal{F}_{P2}$, $\mathcal{F}_{P4}$, and $\mathcal{F}_{G}$, with the exception that the feasible set of goods trades will always be rectangular (the slopes of the boundary lines do not change, as the trades are independent and do not cause externalities over other links in the network).}

By comparing the feasible sets, we can find which trades are feasible in the goods case, but not in the electricity grid ($(q_N,q_S) \in \mathcal{F}_{G}-(\mathcal{F}_{P} \cap \mathcal{F}_{G})$), and also those that are feasible in the electricity grid, but not in a goods trade ($(q_N,q_S) \in \mathcal{F}_{P}-(\mathcal{F}_{P} \cap \mathcal{F}_{G})$). 
Since the power grid is composed of the four individual lines, the final feasible set $\mathcal{F}_{P}$ is the intersection of the sets discussed above. We appeal to two results in convex set theory that are important to note in order to characterize the feasible regions of interest. First, the intersection of an arbitrary collection of convex sets is convex. Second, for any system of simultaneous linear inequalities and equations in $n$ variables, the set of solutions is a convex set in $\mathbb{R}^n$ \citep{rockafellar}. Based on these two results, we state that (i) all the individual feasible sets are convex, and (ii) the intersection of those sets is also a convex set. 

Next, we determine the feasible region for trades for the stylized grid. 

\subsection{Feasible region}
\label{sec:regionfeasibility}

We analyze an asymmetric network where the reactances of the horizontal lines are equal (i.e. $x_{N}=x_{S}=x_H$) and the reactances of the vertical lines are equal  (i.e. $x_{E}=x_{W}=x_V$). The asymmetry of the network is on the transmission line capacities, which are different from each other.

Section~\ref{sec:linefeasibility} shows that the slopes of the parallel lines that limit the feasible regions  $\mathcal{F}_{P2}$ and $\mathcal{F}_{P4}$ are the same. Since the distance between the two boundaries is determined by the line capacity, the intersection between those two sets is the set with the lower capacity. Let $C_V=\min(C_E,C_W)$, then  $\mathcal{F}_{P2} \cap \mathcal{F}_{P4} = \mathcal{F}_{P2}$ if $C_V=C_{E}$ and $\mathcal{F}_{P2} \cap \mathcal{F}_{P4} = \mathcal{F}_{P4}$ if $C_V=C_{W}$. Let us denote that intersection as $\mathcal{F}_{PV}$.  


\begin{figure}[t]
  \subfloat[Indices for boundary lines]{
	\begin{minipage}[c][1\width]{
	   0.5\textwidth}
	   \centering
	   \includegraphics[width=.75\textwidth]{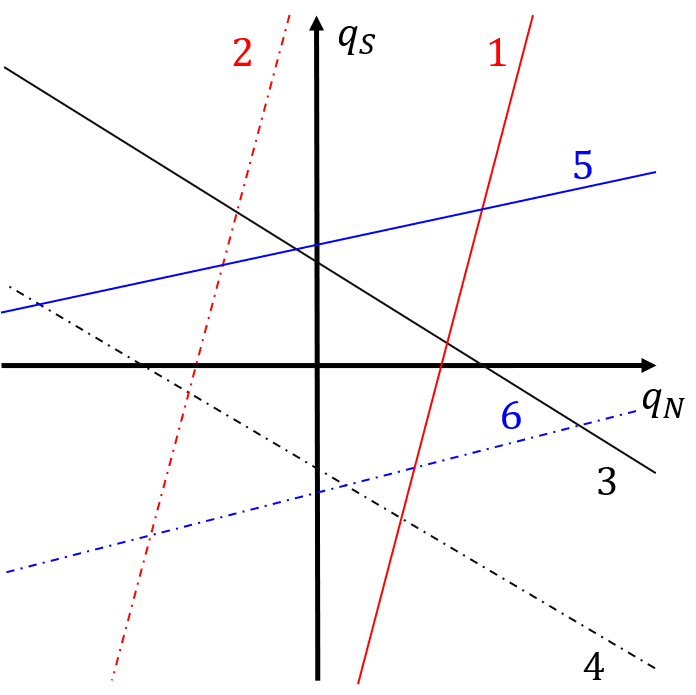}
	   \label{fig:boundsidx}
    \end{minipage}}
 \hfill 	
  \subfloat[Indices for intersections between boundary lines]{
	\begin{minipage}[c][1\width]{
	   0.5\textwidth}
	   \centering
	   \includegraphics[width=.75\textwidth]{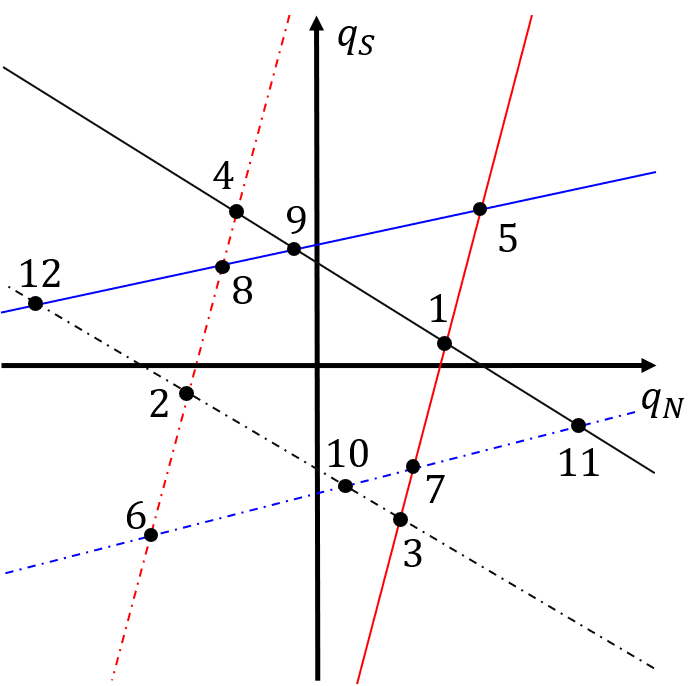}
	   \label{fig:intersectidx}
	\end{minipage}}
\caption{Illustration exemplifying a feasible set for power trades, which is given by the intersection of the feasible trades for each transmission line. The indices for each boundary line (a) and for each intersection (b) are used throughout our analysis in this section. The red lines correspond to the capacity limit of transmission line $\ell=N$. This means that any trade between the two parallel red lines generates a power flow through line $\ell=N$ which satisfies the capacity constraint of this line. Similarly, the black boundaries correspond to the capacity constraint of the vertical line with the lowest capacity, and the blue ones refer to the constraint imposed by line $\ell=S$. The solid lines are the boundaries representing the binding conditions with positive flows, while the dashed lines refer to binding negative flows. Table~\ref{tb:intersec_indices} presents the correspondence between boundary lines and intersection points (e.g. intersection 2 is between boundaries 2 and 4).}
\label{fig:indices}
\end{figure}

To find the feasible trades of the electricity grid as a whole, we need to characterize the intersection of the sets $\mathcal{F}_{P1}$, $\mathcal{F}_{PV}$, and $\mathcal{F}_{P3}$. Since each set is limited by two parallel lines, which correspond to when the power flow constraints are biding, we start by finding all the intersection points between each two lines. This gives us a total of 12 intersection points. 

Let $\mathcal{B}=\{1,...,6\}$ denote the set of indices for the boundary lines of the sets considered. Figure~\ref{fig:indices} illustrates the boundaries and their corresponding indices. The red boundaries correspond to transmission line $\ell=N$ ($\mathcal{F}_{P1}$), the black ones refer to the vertical line with the lowest capacity ($\mathcal{F}_{PV}$), and the blue lines limit the feasible region for $\ell=S$ ($\mathcal{F}_{P3}$). In our convention, odd-numbered lines (solid lines in Figure~\ref{fig:indices}) are the boundaries representing the binding conditions with positive flow, and even-numbered ones (dashed lines in Figure~\ref{fig:indices}) refer to binding negative flows.

For ease of exposition, let $n \in \mathcal{I}=\{1,...,12\}$ denote the indices corresponding to the intersection points. The intersections $a-b$ for lines $a,b\in \mathcal{B}$, which are equivalent to the intersections $b-a$, are numbered as shown in Panel A of Table~\ref{tb:intersec_indices}, and illustrated in Figure~\ref{fig:indices}. The expressions for the coordinates of each intersection point are listed in Panel B of Table~\ref{tb:intersectionsNS}, where $X_r=\frac{x_H}{2(x_H+x_V)}$. 

\begin{table}[tpb]

    \begin{center}
        \begin{tabular}{l c c c c c c c c c c c c}
        \multicolumn{13}{c}{Panel A: Correspondence between intersection points and boundary lines.}\\
        \hline
            $n \in \mathcal{I}$ & 1 & 2 & 3 & 4 & 5 & 6 & 7 & 8 & 9 & 10 & 11 & 12\\ \hline
            $a,b\in \mathcal{B}$ & 1-3 & 2-4 & 1-4 & 2-3 & 1-5 & 2-6 & 1-6 & 2-5 & 3-5 & 4-6 & 3-6 & 4-5 \\
            \hline
        \end{tabular}  
     
    \end{center}
    \centering
    \setlength{\extrarowheight}{10pt}
    \begin{tabular}{ c  c  c  c}
    \multicolumn{4}{c}{Panel B: Intersection points between boundaries of feasible regions $\mathcal{F}_{P1}$, $\mathcal{F}_{PV}$, and $\mathcal{F}_{P3}$}\\
    \hline
        $n $ & $(q_N,q_S)$ &$n$ & $(q_N,q_S)$ \\ \hline
        1 & $\left(C_N+C_V,\frac{(1-X_r)}{X_r}C_V-C_N\right)$ & 2 & $\left(-C_N-C_V,-\frac{(1-X_r)}{X_r}C_V+C_N\right)$ \\ 
        3 & $\left(C_N-C_V,-\frac{(1-X_r)}{X_r}C_V-C_N\right)$ & 4 & $\left(C_V-C_N,\frac{(1-X_r)}{X_r}C_V+C_N\right)$ \\ 
        5 & $\left(\frac{X_rC_S+(1-X_r)C_N}{1-2X_r},\frac{X_rC_N+(1-X_r)C_S}{1-2X_r}\right)$ & 6 & $\left(\frac{-X_rC_S-(1-X_r)C_N}{1-2X_r},\frac{-X_rC_N-(1-X_r)C_S}{1-2X_r}\right)$ \\ 
        7 & $\left(\frac{-X_rC_S+(1-X_r)C_N}{1-2X_r},\frac{X_rC_N-(1-X_r)C_S}{1-2X_r}\right)$ & 8 & $\left(\frac{X_rC_S-(1-X_r)C_N}{1-2X_r},\frac{-X_rC_N+(1-X_r)C_S}{1-2X_r}\right)$ \\ 
        9 & $\left(\frac{(1-X_r)}{X_r}C_V-C_S,C_S+C_V\right)$ & 10 & $\left(-\frac{(1-X_r)}{X_r}C_V+C_S,-C_S-C_V\right)$ \\ 
        11 & $\left(\frac{(1-X_r)}{X_r}C_V+C_S,C_V-C_S\right)$ & 12 & $\left(-\frac{(1-X_r)}{X_r}C_V-C_S,C_S-C_V\right)$ \\
      \hline
    \end{tabular}
\caption{Feasible Power Trades}
\label{tb:intersectionsNS} \label{tb:intersec_indices}
\end{table}

Notice that the intersection points in the left of Table~\ref{tb:intersectionsNS} are the negative of those in the right. Therefore, by symmetry, (1) if an intersection $a-b$ with coordinates $(q_N,q_S)$ is a vertex of the polygon that limits the set of feasible trades, then the intersection with coordinates $(-q_N,-q_S)$ will also be a vertex, (2) the feasible set $\mathcal{F}_{P}$ will be centered at the origin, and (3) if a certain boundary line is a face of the polygon that limits the feasible trades, so is the boundary line parallel to it. Further, since the intersection of a collection of convex sets is also a convex set, the polygon can have at most six faces, corresponding to the six boundary lines of the sets $\mathcal{F}_{P1}$, $\mathcal{F}_{PV}$, and $\mathcal{F}_{P3}$. Note that, for this particular problem, if the polygon had more than six faces, then it would be non-convex. We  exploit that symmetry to determine which intersection points are vertices of the feasible region for this electricity grid.

\begin{theorem}\label{th:feasNS}
Let the vertex $V_{n}$ be the intersection point $n$, as presented in Table~\ref{tb:intersectionsNS}. For a 4-node electricity grid with $x_{N}=x_{S}=x_H$ and $x_{E}=x_{W}=x_V$, the set of feasible power trades $\mathcal{F}_{P}$ can be classified into one of the following 4 cases: 
\begin{itemize}
    \item Case 1: quadrilateral defined by vertices $V_{5}$, $V_{6}$, $V_{7}$ and $V_{8}$.
    
    Network has large capacity $C_V$ on the vertical lines, and thus these constraints will never be binding. This case occurs if the condition below is satisfied.
    \begin{equation}\label{eq:cxratioNS}
        \left(\frac{C_N+C_S}{2C_V}\right)^2\frac{x_H^2}{x_V^2}+ \left(\frac{C_N-C_S}{2C_V}\right)^2\frac{x_H^2}{(x_H+x_V)^2}\leq 1.
    \end{equation}
    
    \item Case 2: quadrilateral defined by vertices $V_{1}$, $V_{2}$, $V_{3}$ and $V_{4}$.
    
    Network has large capacity $C_S$ on the south line, whose capacity constraint will never be binding. This case occurs if
    \begin{equation}
        \frac{2C_S^2(x_H+x_V)^2}{x_H^2+2x_V^2} \geq 2C_N^2+4C_NC_V\frac{x_V}{x_H}+C_V^2\left[1+\frac{(x_H+2x_V)^2}{x_H^2}\right].
    \end{equation}

    \item Case 3: quadrilateral defined by vertices $V_{9}$, $V_{10}$, $V_{11}$ and $V_{12}$.
    
    Network has large capacity $C_N$ on the north line, and thus the constraint on line $i=N$ will never be binding. This case occurs if
    \begin{equation}
       \frac{2C_N^2(x_H+x_V)^2}{x_H^2+2x_V^2} \geq 2C_S^2+4C_SC_V\frac{x_V}{x_H}+C_V^2\left[1+\frac{(x_H+2x_V)^2}{x_H^2}\right]
    \end{equation}
    
    \item Case 4: hexagon defined by vertices $V_{1}$, $V_{2}$, $V_{7}$, $V_{8}$, $V_{9}$ and $V_{10}$.
    
    If none of the conditions above hold, then none of the transmission lines have a large enough capacity relative to others. Then, the hexagon feasible region indicates that the capacity of any of the four lines can become binding. 
\end{itemize}



\end{theorem}
\noindent Proof is in Appendix~\ref{app:proof_th2}.

Figure~\ref{fig:intNS_all} presents one example for each possible case specified in Theorem~\ref{th:feasNS}. The red, blue, and black lines correspond to the individual feasible sets of the transmission line $\ell=N$, transmission line $\ell=S$, and vertical transmission lines, respectively. By fixing the parameters of the horizontal lines and changing only those of the vertical lines, we are able to express three out of the four possible cases. If $C_N>C_S$, those cases will be Cases 1, 3 and 4; if $C_N<C_S$, we may be in cases 1, 2 or 4. If $C_N=C_S$, the conditions for the second and third cases will never hold, and we will reduce the amount of possible cases to two.\footnote{The special case of a symmetric network with $C_N=C_S$ and $C_E=C_W$ is analyzed in Appendix~\ref{sec:symmetric}.}
\begin{figure}
  \subfloat[Case $\mathcal{F}_{P1} \cap \mathcal{F}_{P3} \subset \mathcal{F}_{PV}$, with  $C_N= 150$, $C_V = 100$, $C_S = 100$, $x_H = 0.2$, $x_V = 0.3$]{
	\begin{minipage}[c][0.9\width]{0.475\textwidth}
	   \centering
	   \includegraphics[width=0.9\textwidth]{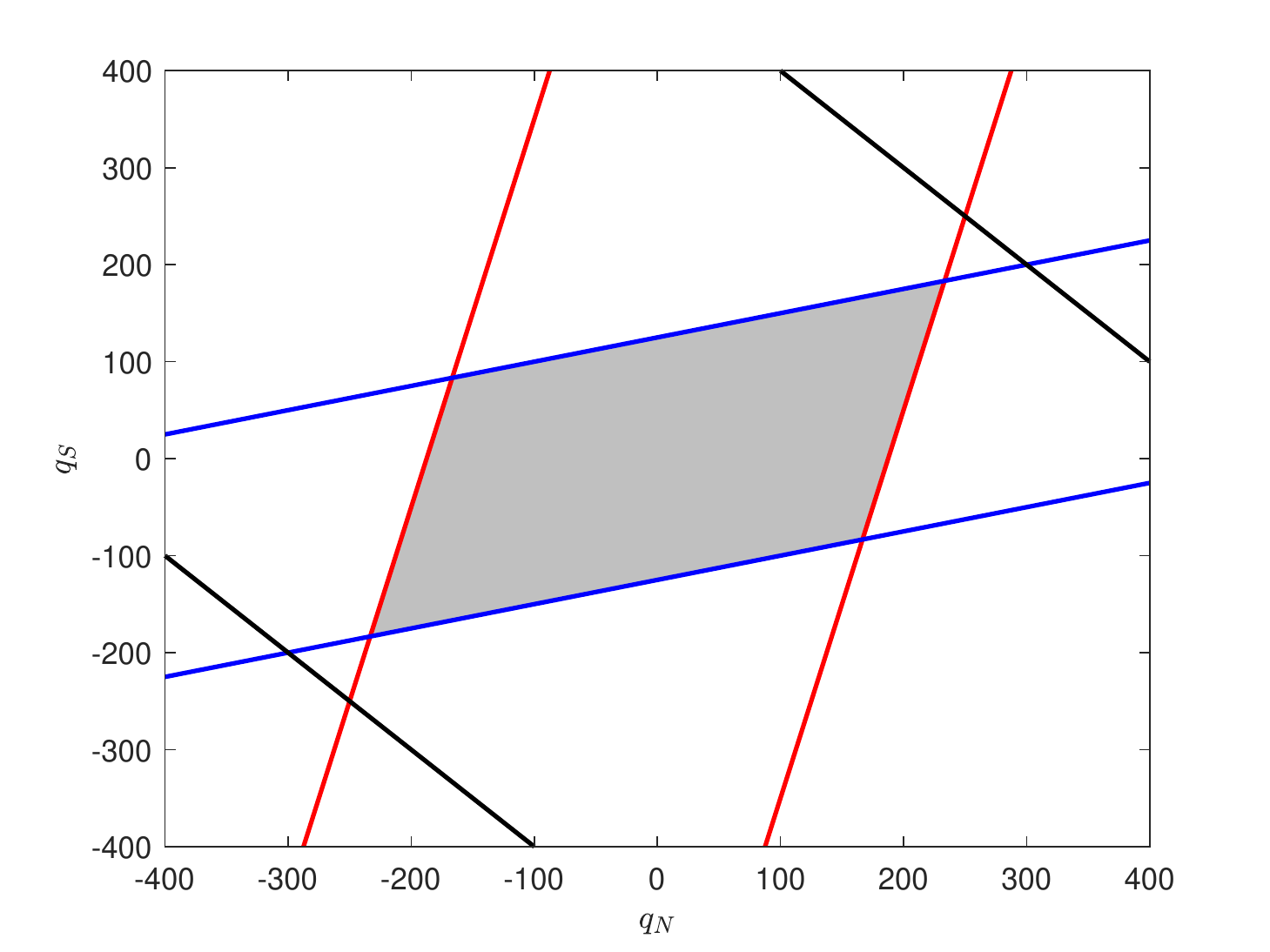}
	   \label{fig:intNSCase1}
	\end{minipage}}
    \hfill
    \subfloat[Case $\mathcal{F}_{P1} \cap \mathcal{F}_{PV} \subset \mathcal{F}_{P3}$, with  $C_N= 100$, $C_V = 45$, $C_S = 150$, $x_H = 0.2$, $x_V = 0.1$.]{
	\begin{minipage}[c][0.9\width]{0.475\textwidth}
	   \centering
	   \includegraphics[width=0.9\textwidth]{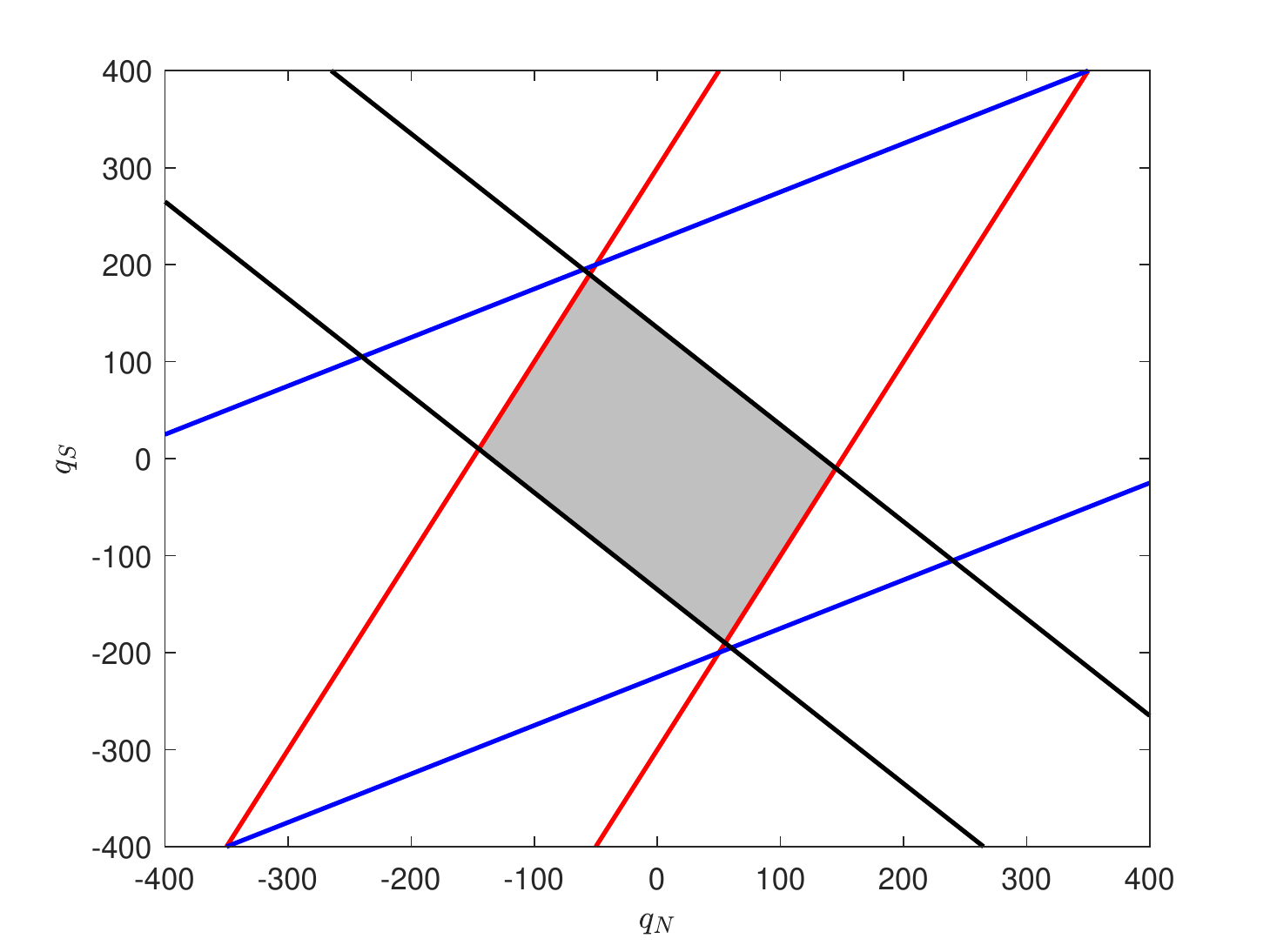}
	   \label{fig:intNSCase2}
	\end{minipage}}
    \newline	
    \subfloat[Case $\mathcal{F}_{PV} \cap \mathcal{F}_{P3} \subset \mathcal{F}_{P1}$, with  $C_N= 150$, $C_V = 45$, $C_S = 100$, $x_H = 0.2$, $x_V = 0.1$]{
	\begin{minipage}[c][0.9\width]{0.475\textwidth}
	   \centering
	   \includegraphics[width=0.9\textwidth]{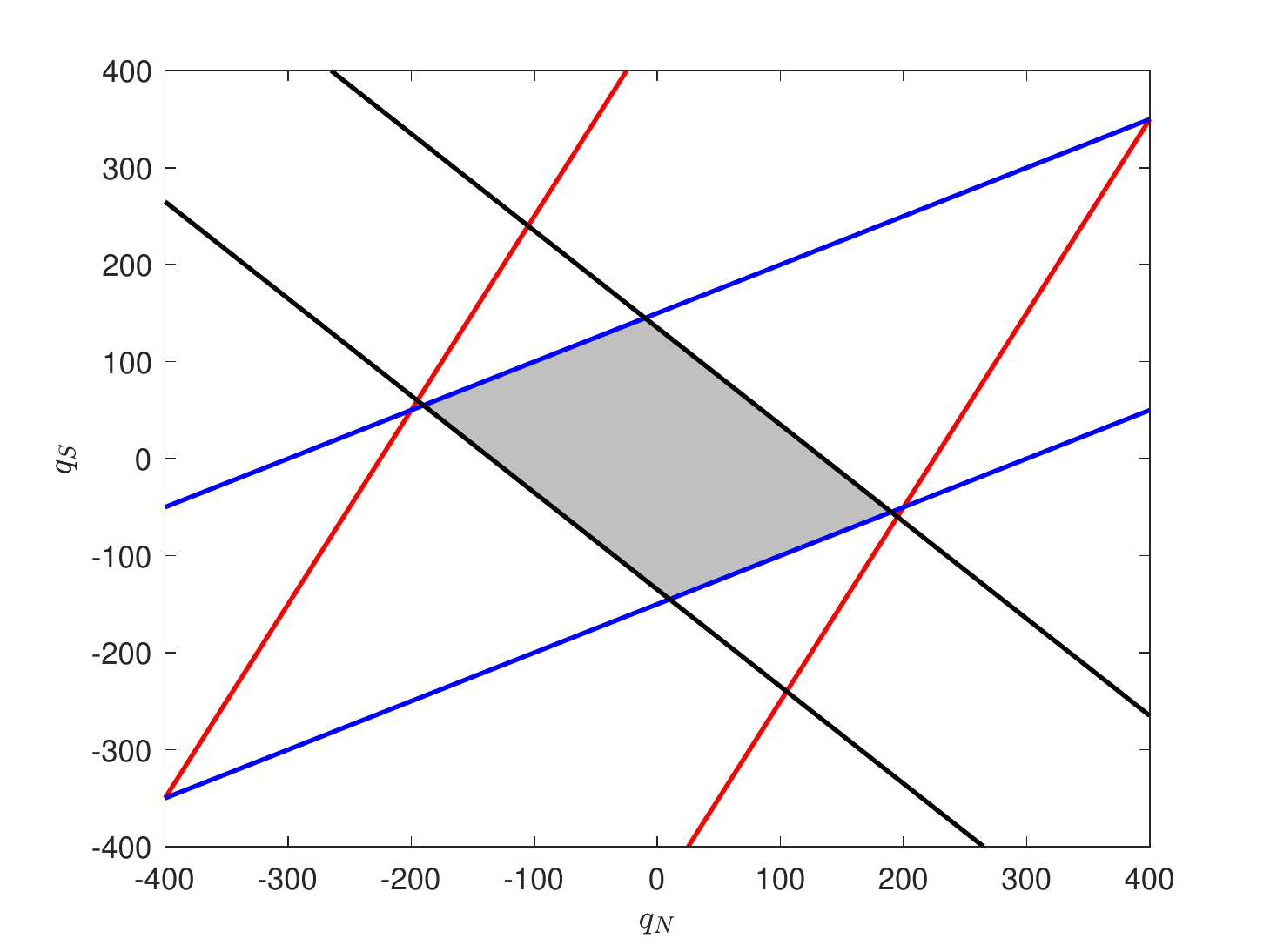}
	   \label{fig:intNSCase3}
	\end{minipage}}
	\hfill
	\subfloat[Hexagon Case, with  $C_N= 150$, $C_V = 100$, $C_S = 100$, $x_H = 0.2$, $x_V = 0.1$]{
	\begin{minipage}[c][0.9\width]{0.475\textwidth}
	   \centering
	   \includegraphics[width=0.9\textwidth]{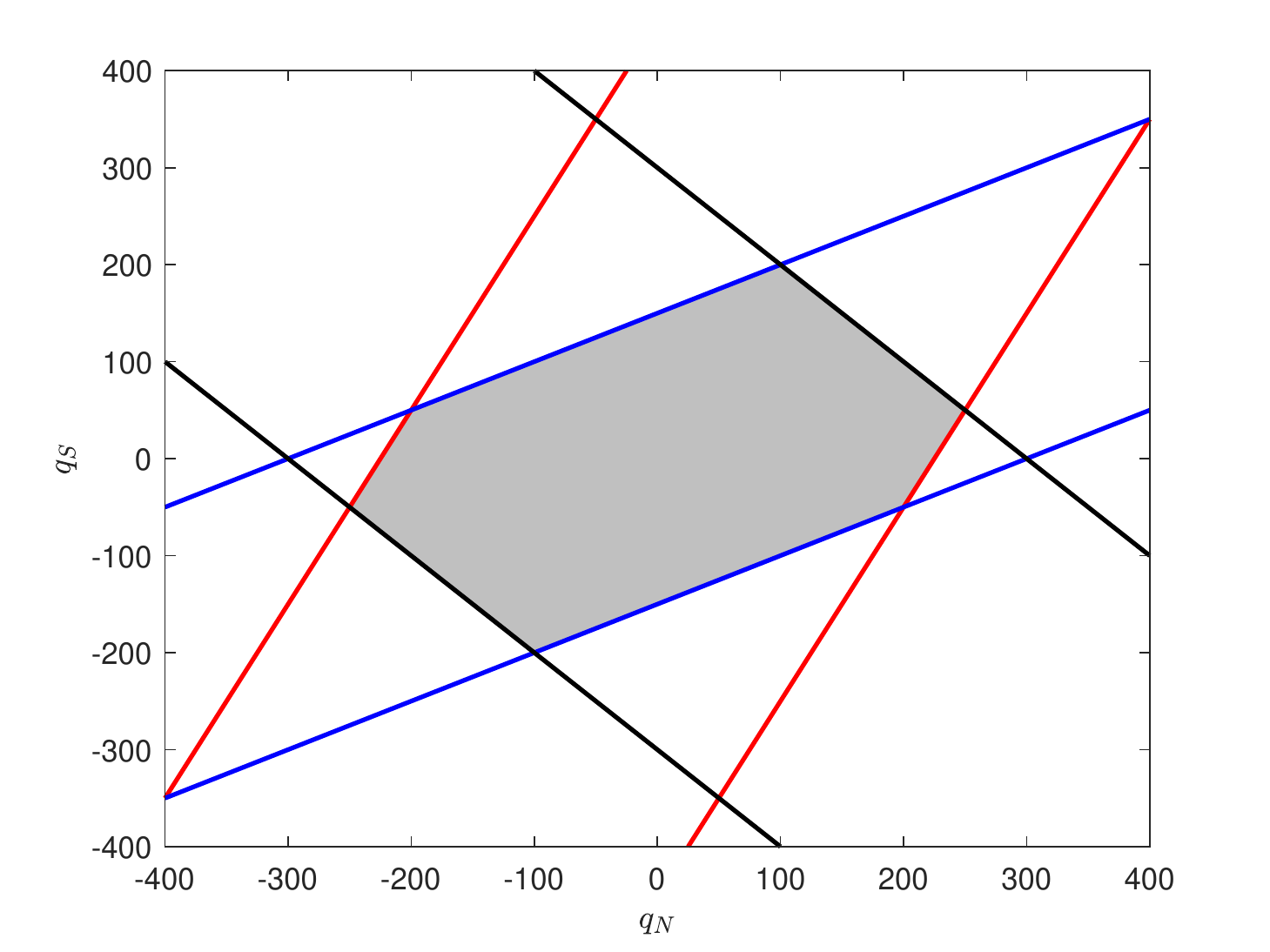}
	   \label{fig:intNSCase4}
	\end{minipage}}
\caption{Illustration of the possible feasible regions for power trades as characterized in Theorem~\ref{th:feasNS}. In (a), network parameters are such that the feasible set satisfies the conditions of Case 1. From this example, we can enter Case 4 (hexagon) by decreasing the vertical reactance $x_V$, which will decrease the slope of the red lines and increase the slope of the blue lines. The new intersection points of these lines with the y-axis are such that the sets intersect as shown in (d). Note that we could also have achieved this case by fixing all parameters used in (a), except for $C_V$, which, if decreased, would intercept $\mathcal{F}_{P1} \cap \mathcal{F}_{P3}$ to form a hexagon. From (d), we are led to Case 3 if the capacity $C_V$ is lowered, as depicted in (c). This change only affects the black lines, which will become less wide. From (c), if we switch the capacity values of $C_N$ and $C_S$, we get to (b), which is an example of network configuration for which Case 2 holds.}
\label{fig:intNS_all}
\end{figure}

We can also determine the intersection of the set of feasible power trades with that of feasible goods trades. Recall that goods trades are independent of one another, in that they do not cause externalities to other trades. For that reason, the feasible region of non-power trades is a rectangle, as presented in Figure~\ref{fig:feasiblesets}, and its vertices are $V_a=(-C_N,C_S)$, $V_b=(C_N,C_S)$, $V_c=(C_N,-C_S)$, and $V_d=(-C_N,-C_S)$.

\section{Challenges for the European Power Grid}
\label{sec:eurogrid}

As Europe has integrated politically and economically, the electricity markets of Europe have also been integrated. The benefits  of electricity market integration are similar to those of market integration in general. Seminal work includes \cite{Krugman:79,Krugman:80,Krugman:81,Krugman:91} who show that gains from trade due to increasing returns to scale, i.e. cost efficiency can occur even between countries with similar preferences, technology, and factor endowments. This is applicable in the case of electricity grid in Europe where the good is homogeneous and is traded between countries which are similar in levels of technology.  

However, as we argue in this paper, electricity market integration can also lead to a unique externality. Depending upon the direction and magnitude of power trades, the transmission capacity available for non-trading nodes may decline. This can lead to transmission constraints on the grid, reducing or even eliminating the gains from integrating the network. Such constraints do not bind if power trades are concentrated locally. This is because historically, most production and consumption of power is co-located. 

The introduction of renewable sources of electricity disrupts this equilibrium.  Power production from renewable sources may not be evenly distributed over the grid. For example, the northern European region has made great strides in producing renewable energy and this energy has created large power trades  over the electrical grid. Given the confluence of these two factors, an integrated electrical grid and large power trades over the grid, Europe provides an ideal setup for us to determine the magnitude of power trade externalities on the electrical grid.

\subsection{Calibration}
\label{sec:calib}

 
We divide the European grid into four meaningful regions, as shown in Figure~\ref{fig:europemap}. Western Europe includes France, Spain, and Portugal. The northern region includes Germany and Scandinavia in our exercise. Eastern Europe includes countries such as Poland and Hungary. Southern Europe includes includes Italy and Greece. We observe that nodes NE and SW in the European grid are connected, unlike our four-node model. For the purpose of our analysis, we can eliminate this extra connection to evaluate only the remaining trades.\footnote{Following \cite{shi}, we can remove the extra line such that the remaining flows in the network remain the same. If the trade flows from node $NE$ to node $SW$, this is done by adding a load in node $NE$ that consumes the flow in the extra line, and adding a generator in node $SW$ that generates that same amount. If the trades flow in the opposite direction, we only need to switch the position of the load and generator added.}
\begin{figure}[h]
    \centering
    \includegraphics[width=0.75\linewidth]{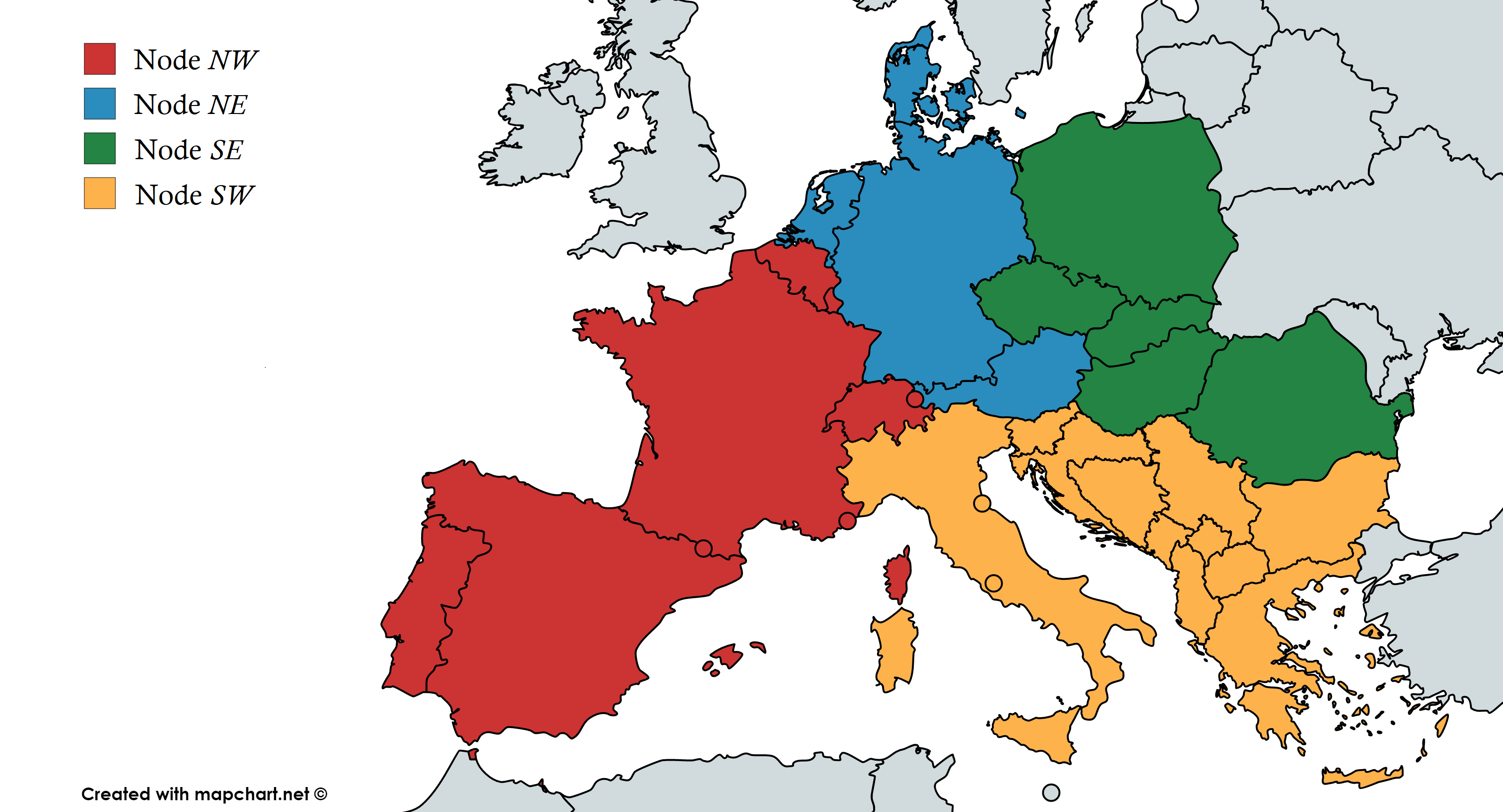}
    \caption{European grid division into four nodes}
    \label{fig:europemap}
\end{figure}

The parameters of the four-node network considered comprise line reactances $x_\ell$ and capacities $C_\ell$ for each line $\ell \in \{N,E,S,W\}$. To estimate the reactance values, we use the publicly available data from the European Network of Transmission System Operators for Electricity (ENTSO-E) Transparency Portal.\footnote{Available at \url{https://transparency.entsoe.eu/transmission-domain/physicalFlow/show}} The portal collects and publishes electrical generation, transportation and consumption data for Europe. The dataset of interest is the hourly cross-border power flow data in Europe.

\textbf{Power flows:} We begin by picking a 50-day time period (March 12th -- April 30th, 2019) with typical amounts of cross-border trades among the European countries. This period of spring is recent and avoids concerns that large spikes in consumption due to cooling or heating needs are driving the results. For each  of the significant regions in Figure~\ref{fig:europemap}, we calculate average power flows that cross the border between regions. Kirchhoff's law determines that, at each node $i$, the sum of power inflows must equal the sum of power outflows. From this, we can use the average power flows computed to determine the average net injections $q_i$ for this 50-day period in Spring 2019. The calculated values  are presented in Figure~\ref{fig:entsoe}. For the time period considered, we note that the region containing Germany was a big power exporter. 


\begin{figure}[h]
    \centering
    \includegraphics[width=0.5\linewidth]{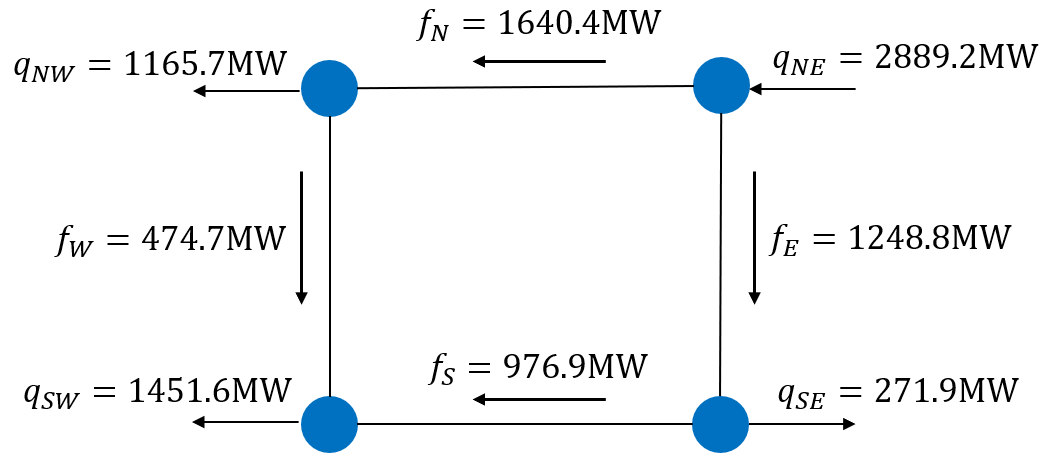}
    \caption{Average power flows through each line and power injections in each node calculated using the ENTSO-E data. These values are used to estimate the reactances of the four-node network.}
    \label{fig:entsoe}
\end{figure}

We aim to estimate the line reactances using the power flows and injections computed. Recall that the power flows satisfy the relation $\bm{f^p}=\mathbf{Sq} $ where the shift-factor transformation matrix $\mathbf{S}$ is a function of the reactances $x_\ell$, as presented in Appendix \ref{app:Smatrix}. Let $\overline{\bm{f}}$ and $\overline{\mathbf{q}}$ denote the vectors of average power flows and net injections calculated using the ENTSO-E data. 

\textbf{Effective Reactances:} We calculate the effective reactances $\mathbf{x}$ that lead to the power flows $\overline{\bm{f}}$ calculated from data, given the net injections $\mathbf{\overline{q}}$ also obtained from the data. To obtain effective reactances, we solve:
\begin{equation*}
    \underset{\mathbf{x}}{\min}~\lVert \overline{\bm{f}} - \mathbf{S(x)\overline{q} \rVert}.
\end{equation*}
We estimate that the symmetric reactance vector $\mathbf{x} = (x_N,x_E,x_S,x_W) = (0.5621, 0.4818, 0.5621, 0.4818)$ approximates the power flows with error on the order of $10^{-5}$.

\textbf{Line Capacities:} Next, we estimate the capacities of the four lines in the network. The ENTSO-E Transparency Portal does not provide this information. Hence, we use data from \citep{hutcheon}. This dataset lists all transmission lines that cross the boundaries of the regions analyzed, along with their capacities.  Figure~\ref{fig:europezones} from \cite{hutcheon} shows that there is significant power trade over large distances in the European grid. The red partitions are ours, based on the partition of Europe into four nodes (Figure~\ref{fig:europemap}). 

\begin{figure}[!ht]
    \centering
    \includegraphics[width=0.7\linewidth]{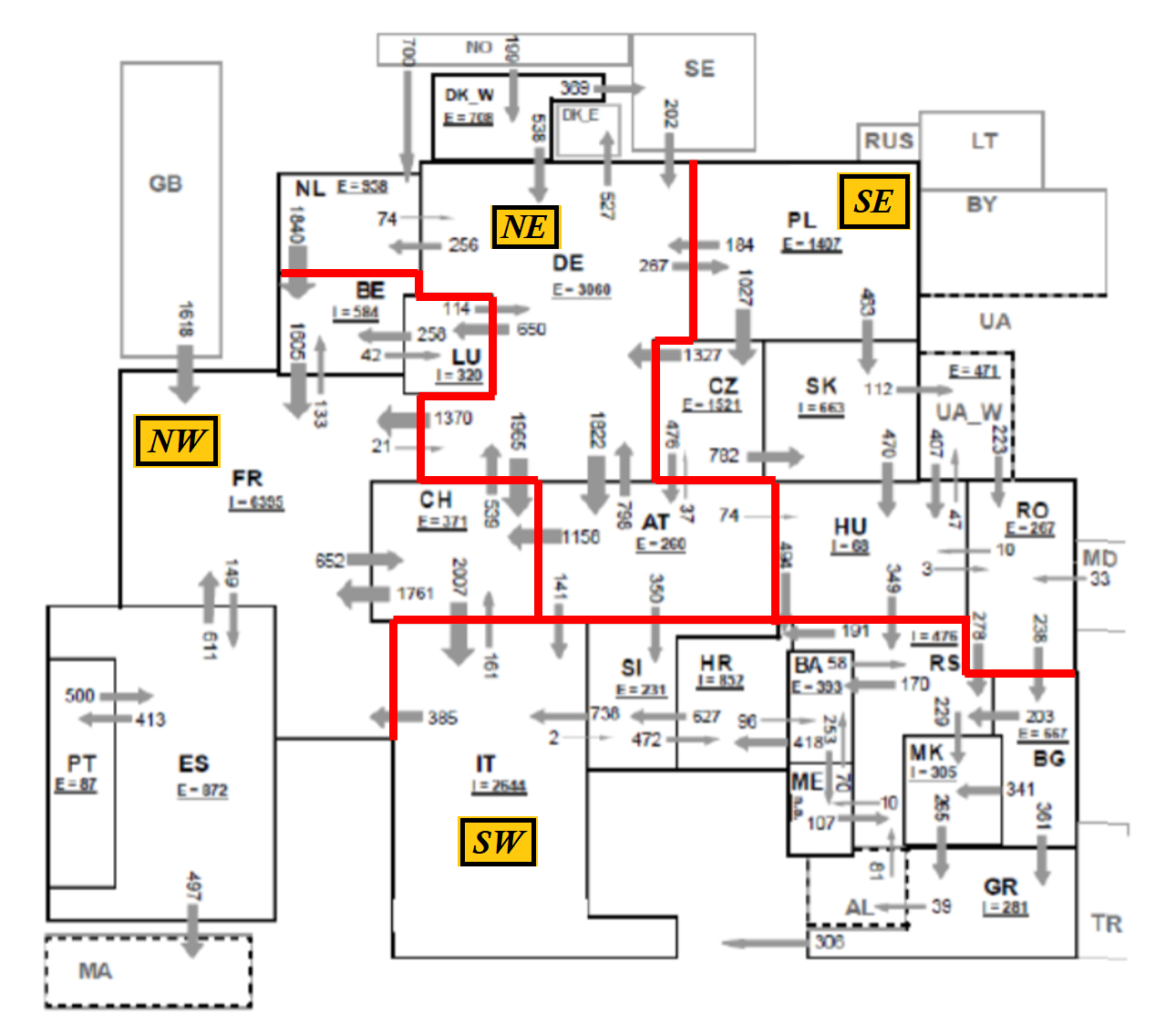}
    \caption{Power flow data for the continental European transmission network. The set of trades derived from this dataset are used in the analysis presented in Section~\ref{sec:simulation}.}
    \label{fig:europezones}
\end{figure}

Assuming that all the transmission lines across two different zones are in parallel, we add the capacities of all lines which connect each two zones to find the equivalent capacities below:
\begin{equation*}
        C_N=18860MW,~C_E = 9796MW,~C_S = 8021MW, ~C_W = 4880MW
\end{equation*}


\textbf{Power Injections:} The abstraction into four-subdivisions allows us to note that nodes $NW$ (Western Europe) and $SW$ (Southern Europe) are net consumers. In contrast, nodes $NE$ (Northern Europe) and $SE$ (Eastern Europe) are net producers.

Data from \cite{hutcheon} further allows us to calculate net injections, $q_{NW}=-4846MW$, $q_{NE} = 4624MW$, $q_{SE} = 3042MW$, $q_{SW} = -2820MW$, where subscripts refer to each node in our model, which correspond to the regions of Europe. For each region, these values are calculated as follows. First, we add all the power flows which leave each region and, from this total outflow, we subtract the sum of all power flows entering this region. Thus, net suppliers have a positive net injection, while net consumers have a negative one. 

We observe that these trades are close to a symmetric case, and thus can be approximated to illustrate a case in which Western Europe trades with Northern Europe, while Eastern Europe trades with Southern Europe. In the discussions that follow, we consider the following approximation for the trades:
\begin{align*}
    q_{NW} = -q_{NE} = q_N &= -\frac{4846+4624}{2}MW = -4735 MW \\
    q_{SE} = -q_{SW} = q_S &= \frac{3042+2820}{2}MW = 2931 MW
\end{align*}

\subsection{Network Congestion Externalities}
\label{sec:simulation}
\label{sec:caseexternality}


\textbf{Positive or negative externality:} In the stylized framework, the $NW$ node (Western Europe, including France)  is a consumer of the power produced by the $NE$ node (Northern Europe, including Germany). At the same time, node $SE$ (Eastern Europe) produces power that is consumed by node $SW$ (Southern Europe, including Italy). Theorem~\ref{th:sigmaf} states that if the ratio of the trades is greater than the ratio of effective reactances, i.e.  
\begin{equation*}
    |q_N/q_S|=1.6155 > |x_{N}/x_{S}|=1,
\end{equation*}
then the power flows lead to a tighter bound in transmission line $\ell=S$ (lower horizontal). That is, this line is subject to a negative externality imposed by the flows generated by the trade between the upper two nodes. 

\textbf{Quantifying the average externality:} The externality identified above can be quantified by decomposing the power flows using Lemma~\ref{lem:powerflows}.
\begin{align*}
    \bm{f^p}&=\bm{f^g} + \mathbf{S_\sigma}~\mathbf{q}\\
    &= \begin{bmatrix}
        q_N \\ 0 \\ q_S \\ 0
        \end{bmatrix}+
    \frac{1}{\sum x_{\ell}}\begin{bmatrix}
        -(x_{N}+x_{E}+x_{S}) & -(x_{E}+x_{S}) & -x_{S} & 0\\
        x_{W} & x_{N}+x_{W} & -x_{S} & 0\\
        x_{W} & x_{N}+x_{W} & -x_{S} & 0\\
        x_{N}+x_{E}+x_{S} & x_{E}+x_{S} & x_{S} & 0
    \end{bmatrix}
    \begin{bmatrix}
        q_N \\ -q_N \\ q_S \\ -q_S\\
    \end{bmatrix}\\
    &= \begin{bmatrix}
        q_N \\ 0 \\ q_S \\ 0
        \end{bmatrix}+
        \frac{1}{\sum x_{\ell}}\begin{bmatrix}
        -x_Nq_N-x_Sq_S \\ -x_Nq_N-x_Sq_S \\ -x_Nq_N-x_Sq_S \\ x_Nq_N+x_Sq_S\\
    \end{bmatrix},
\end{align*}
where the second term denotes the $\sigma$-flows characterized in Theorem~\ref{th:sigmaf}. Using the calibrated values from Section~\ref{sec:calib}, the decomposition is as follows:
\begin{equation*}
    \bm{f^p}= \begin{bmatrix}
    -4735 \\ 0 \\ 2931 \\ 0
    \end{bmatrix}+
    \begin{bmatrix}
    485.7 \\ 485.7 \\ 485.7 \\ -485.7\\
    \end{bmatrix} = \begin{bmatrix}
    -4249.3 \\ 485.7 \\ 3416.7 \\ -485.7\\
    \end{bmatrix}
\end{equation*}
The above decomposition suggests that the $\sigma$-flows have a magnitude of $485.7$ MW and they are in the clockwise direction, which is in accordance with the characterization presented in Theorem~\ref{th:sigmaf} for these flows. An illustration of the flow decomposition performed is presented in Fig.~\ref{fig:fdecomposition}.
    \begin{figure}[h!]
        \centering
        \includegraphics[width=0.8\linewidth]{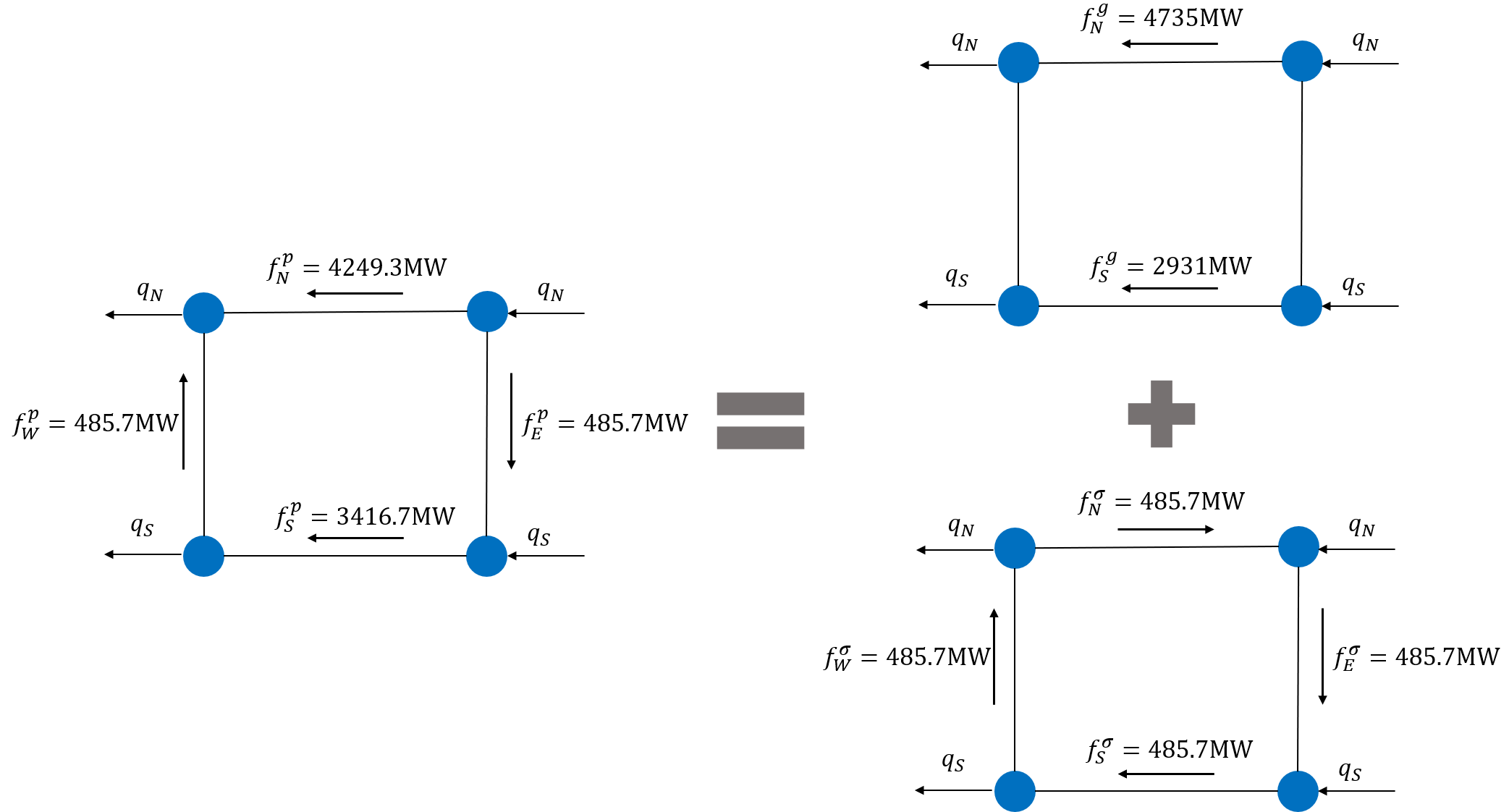}
        \caption{Power flow decomposition into goods component and clockwise $\sigma$-flows.}
        \label{fig:fdecomposition}
\end{figure}

When both trades are performed simultaneously in the network,\footnote{Appendix~\ref{sec:further} further decomposes the externality by separating the portions caused by the North-North and South-South trade if they were individually present.} the South transmission line faces a negative externality in the form of a cyclic clockwise flow that decreases the capacity available for nodes $SE$ and $SW$ to trade. This extra flow is $\frac{485.7}{4735} = 10.26\%$ of the trade between nodes $NW$ and $NE$. 

Conversely, this clockwise flow \textit{expands} the available capacity of the North transmission line, and therefore trades between nodes $NW$ and $NE$ face a positive externality. This externality is in the form of a counter-flow that has $     \frac{485.7}{2931} = 16.57\% $ of the magnitude of the trade between nodes $SE$ and $SW$.

\textbf{Externality due to marginal trade:} Keeping the South-South trade fixed, each unit added to the 4735 MW North-North trade increases the extra clockwise flow by 0.2692 units. This value can be found by calculating the extra flow $\mathbf{S_\sigma}~\mathbf{q}$ for the unit additional trade $(q_N,q_S)=(1,0)$. This result suggests that the value of the transit fare should be 26.92\% for additional North-North trades. The fare can be utilized to invest in the power infrastructure in the  countries represented in the Southern region.




\subsection{Trade Feasibility}

The previous section calculates the network externalities in Europe when large power flows take place. This section estimates how the set of feasible trades change given that power has to obey physical laws. Using the estimated network parameters and power trades, we can analyze the feasible region for power trades and how it overlaps the set of feasible goods trades.

Figure~\ref{fig:overlapping}  shows that there are trades $(q_N,q_S)$ which are feasible in the power case but not in the goods case, and vice-versa. The area of feasible trades over the grid if electricity did not have to obey physical laws is 4\% larger than the feasible trade region for power trades. This estimate is conditional on the calibrated parameters and suggests that we have a set contraction when we move from goods to power trades. Note that this is a contraction in the feasible trades which does not weigh trades by probabilities. The previous section calculates externalities by considering trades that take place.

For the set of trades as calibrated,  none of the transmission capacities are binding. This is the case for the goods and for the power trade networks. The present trade realization is in the interior of the intersection of the feasible sets for these two cases, as shown in the left panel of Figure~\ref{fig:overlapping}. 

\begin{figure}[h!]
    \centering
    \includegraphics[scale=0.5]{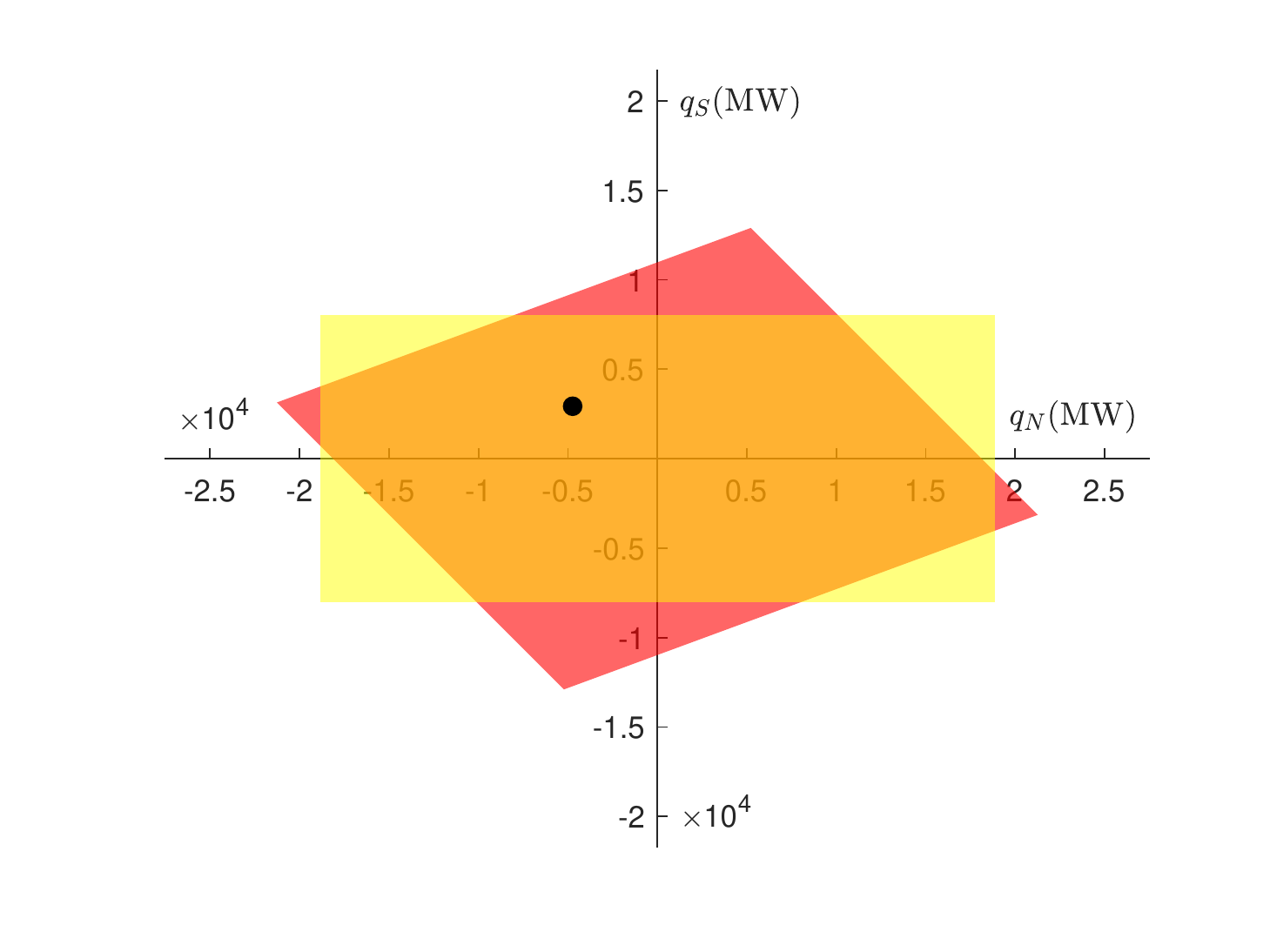}
    \includegraphics[scale=0.5]{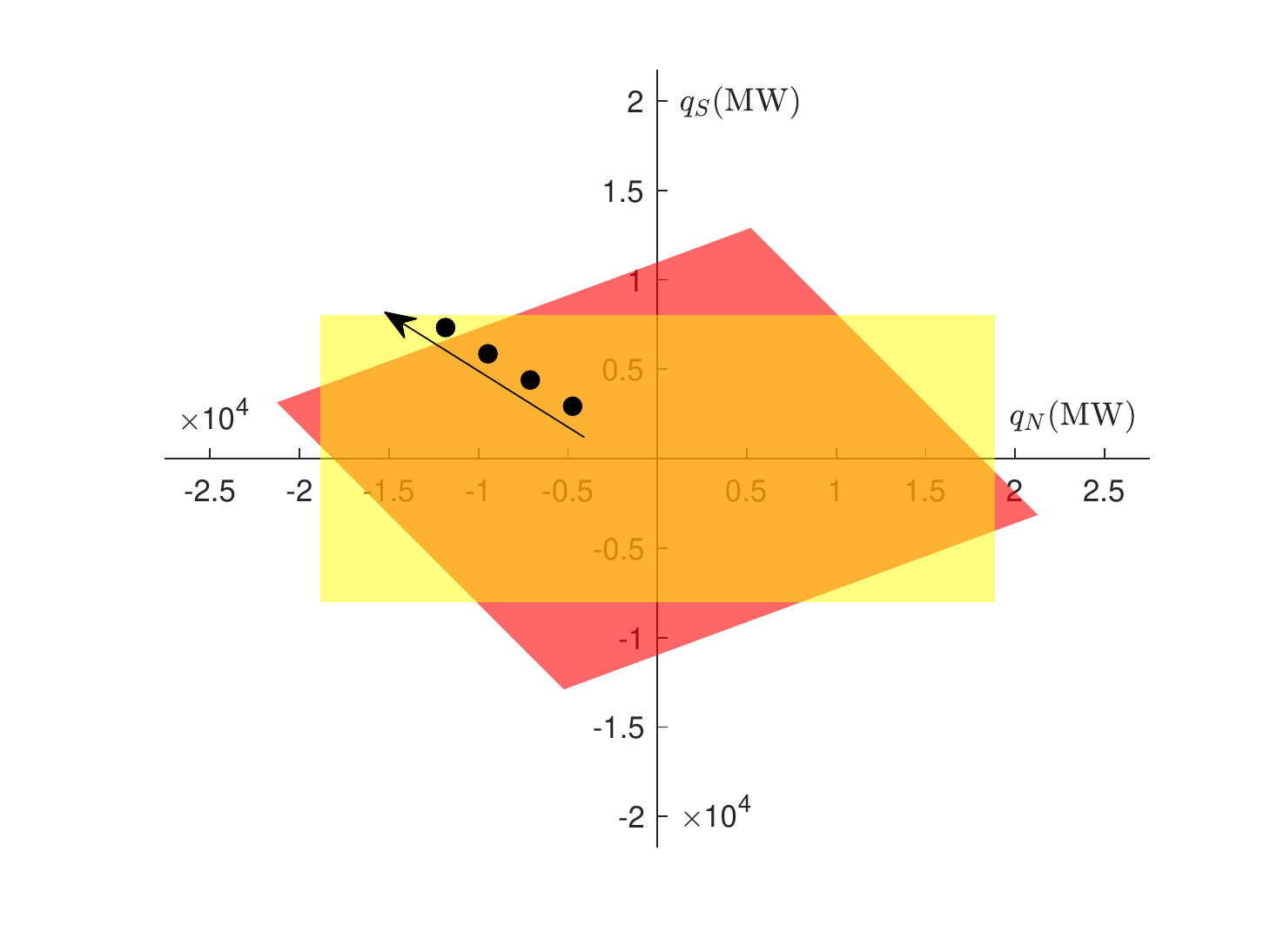}
    \caption{The left panel shows the overlap of feasible set for power (red) and goods (yellow) trades. The black dot marks the trade pair $(q_N,q_S)$ for Europe. The right panel shows the evolution of trade $(q_N,q_S)$ as both North-North and South-South renewable power trades increase at the same rate.}
    \label{fig:overlapping}
\end{figure}


Finally,  use the calibrated network parameters to analyze the effect of increasing the North-North and the South-South trades. 
If both trades increase at the same rate, then the feasibility of trades will evolve as shown in the right panel of Figure~\ref{fig:overlapping}. Figure~\ref{fig:renewabletrends} suggests that renewable energy production has increased significantly in recent years, more than doubling over the last 20 years, suggesting an annual rate of 5\% or more. 

Given this trend, we next try to understand when the capacity constraints will bind in the calibrated power network. The 
southernmost black dot in the right panel of Figure~\ref{fig:overlapping} is the same as in the left panel, corresponding to the baseline case with trades $q_N=-4735MW$ and $q_S=2931MW$. The trade pair $(q_N,q_S)$ scales up by a factor $\alpha>1$. As these power trades increase at the same rate, we move towards the boundary of the set of feasible trades. Ultimately, when the multiplying factor reaches $\alpha=2.5$, the combined trades leave the feasible set for power trades, but they are still feasible in the goods case. This corresponds to the northernmost black dot in the right panel of Figure~\ref{fig:overlapping}. At this point, the power flows generated by the trade pair $(q_N,q_S)=(-11837.5,7327.5)$ would violate the capacity constraint of the South transmission line and thus could not be allowed. Note that the feasibility in the goods case holds for this trade pair because $|q_N|<C_N$ and $|q_S|<C_S$. However, if both trades keep increasing at the same rate, the goods trades also become infeasible due to a capacity violation in the South line, when $|q_S|>C_S=8021MW$.  

The calculations thus show that the power trades have 8.7\% lower grid capacity that goods trades, which is an economically significant effect of physical laws. Our work suggests that newer investment can take the network externalities into account and expand grid transmission capacity on other lines as well (in this case, the South line) to avoid grid congestion. 

To provide additional background, the trade flow on the North-North line is 8\% of the production capacity of Germany in 2019, suggesting a slack of approximately 12\% on the line at present.\footnote{For reference, in 2019, Germany produced an average of 58,789 MW over the year, with renewable power being approximately half of the production. See Public Net Electricity Generation in Germany 2019  \url{https://www.ise.fraunhofer.de/content/dam/ise/en/documents/News/0120_e_ISE_News_Electricity\%20Generation_2019.pdf}.} The trend line of growth of energy production suggests that the capacity constraint may be reached in approximately 2--3 years. Thus, additional investment in the grid maybe needed relatively soon.





\section{Conclusion}
\label{sec:conc}

As the production and consumption patterns of electricity change, the electric grid needs to evolve as well. As a higher fraction of electricity is produced by renewable sources, co-location of production and consumption becomes less feasible. Hence, the grid needs to routinely transport large amounts of power across the network. 

Our paper shows that such power trades produce unique and potentially negative externalities that are not present when trading other goods. Trade between one energy producer-consumer can restrict trades between other energy producer-consumer pairs over the grid due to physical laws. We characterize the set of feasible trades, and decompose the additional externality in power grids over trade in goods. 

We apply our analysis to the European power grid and find that Europe is facing negative externalities of power trades. For each additional unit of power traded between the Northern countries, the transmission capacity for the Southern European grid decreases by approximately 27\% of the traded power. Thus, investment in electric grids should be done across the network and not piecemeal. A transit fare can provide financial resources for investment to parties that are facing externalities.


\bibliographystyle{jof}
\bibliography{bib}

\clearpage

\appendix
\setcounter{table}{0}
\renewcommand{\thetable}{\Alph{section}.\Roman{table}}
\setcounter{figure}{0}
\renewcommand{\thefigure}{\Alph{section}.\arabic{figure}}

\begin{center}
    {\LARGE{\textbf{Appendix: For Review and Online Publication Only}}}
\end{center}

\section{Shift-matrix computation}\label{app:Smatrix}

The shift-factor matrix $\mathbf{S}$ linearly maps the vector of power injections $\mathbf{q}$ to a vector of power flows $\bm{f}$ through the transmission lines of the electricity grid. This matrix depends on the admittances of the transmission lines and incorporates Kirchoff's laws. To determine the elements of this matrix, we first introduce the relationships below:
\begin{align}
    \mathbf{\theta} &= \mathbf{B}^{-1}\mathbf{q} \label{eq:B}\\
    \bm{f} &= \mathbf{H}\mathbf{\theta} \label{eq:H}
\end{align}
In \eqref{eq:B}, the matrix $\mathbf{B}$ is the bus admittance matrix (excluding the reference bus), which relates the phase angles $\mathbf{\theta}$ and the power injections $\mathbf{q}$. In the expression \eqref{eq:H}, $\mathbf{H}$ maps the phase angles in each node to the line flows $\bm{f}$ along the transmission lines. 

Let $x_{ij}$ be the reactance of the transmission line between nodes $i$ and $j$. The elements of the bus admittance matrix $\mathbf{B}$ are given by
\begin{align}
    B_{ii} &= \sum_j \frac{1}{x_{ij}}\\
    B_{ij} &= -\frac{1}{x_{ij}}, ~i\neq j.
\end{align}
In case two nodes $i$ and $j$ are not directly connected, the terms corresponding to this pair are zero. Further, to calculate the power flows, we eliminate the row and column of $\mathbf{B}$ that refer to the reference bus, so that the matrix becomes non-singular. The phase angle-line flows relationships in \eqref{eq:H} are given by
\begin{equation}
    f_{ij} = \frac{\theta_i-\theta_j}{x_{ij}},
\end{equation}
where $f_{ij}$ is the flow from node $i$ to node $j$ and $\theta_i$ is the voltage angle in node $i$. Writing that for all transmission lines $\ell$, we find that the mapping $\mathbf{H}$ is given as
\begin{align}
    H(\ell,i) &= \frac{1}{x_{ij}}\\
    H(\ell,j) &= -\frac{1}{x_{ij}}\\
    H(\ell,m) &= 0~~\forall m \neq i, m\neq j.
\end{align}

Since we are interested in the relationship between power injections and line flows, we can substitute \eqref{eq:B} in \eqref{eq:H} to find our shift matrix $\mathbf{S}$
\begin{equation}\label{eq:pflow}
    \bm{f} = \mathbf{HB}^{-1}\mathbf{q} = \mathbf{S}\mathbf{q}.
\end{equation}
Because the power can flow both ways in the line, the values in $\bm{f}$ may be positive or negative. 

\section{Proof of Lemma~\ref{lem:powerflows}}\label{app:proof_lem1}

Let $\sum x_{\ell}:=(x_N+x_E+x_S+x_W)$. The power flows $f^p_{\ell}$ through each link $\ell$ in the electricity grid are
\begin{align}
    \nonumber \bm{f^p} =
    \begin{bmatrix}
    f^p_{N} \\ f^p_{E} \\ f^p_{S} \\ f^p_{W}
    \end{bmatrix}&=
    \frac{1}{\sum x_{\ell}}\begin{bmatrix}
    x_W & -(x_E+x_S) & -x_S & 0\\
    x_W & x_N+x_W & -x_S & 0\\
    x_W & x_N+x_W &x_N+x_W+x_E & 0\\
    x_N+x_E+x_S & x_E+x_S & x_S & 0
\end{bmatrix}
\begin{bmatrix}
    q_N \\ -q_N \\ q_S \\ -q_S
\end{bmatrix}\\
    \nonumber &=
    \begin{bmatrix}
       1 & 0 & 0 & 0\\
       0 & 0 & 0 & 0\\
       0 & 0 & 1 & 0\\
       0 & 0 & 0 & 0
    \end{bmatrix}
    \begin{bmatrix}
        q_N \\ -q_N \\ q_S \\ -q_S
    \end{bmatrix}+
    \frac{1}{\sum x_{\ell}}\begin{bmatrix}
    x_W-\sum x_{\ell} & -(x_E+x_S) & -x_S & 0\\
    x_W & x_N+x_W & -x_S & 0\\
    x_W & x_N+x_W &x_N+x_W+x_E-\sum x_{\ell} & 0\\
    x_N+x_E+x_S & x_E+x_S & x_S & 0
\end{bmatrix}
\begin{bmatrix}
    q_N \\ -q_N \\ q_S \\ -q_S
\end{bmatrix}\\
    \nonumber &=
   \begin{bmatrix}
    f^g_{N} \\ f^g_{E} \\ f^g_{S} \\ f^g_{W}
    \end{bmatrix}+
    \frac{1}{\sum x_{\ell}}\begin{bmatrix}
    -(x_N+x_E+x_S) & -(x_E+x_S) & -x_S & 0\\
    x_W & x_N+x_W & -x_S & 0\\
    x_W & x_N+x_W & -x_S & 0\\
    x_N+x_E+x_S & x_E+x_S & x_S & 0
\end{bmatrix}
\begin{bmatrix}
    q_N \\ -q_N \\ q_S \\ -q_S
\end{bmatrix}\\
     \nonumber &= \bm{f^g} + \mathbf{S_\sigma}~\mathbf{q},
\end{align}
which is precisely the expression presented in the lemma.

\section{Proof of Theorem~\ref{th:sigmaf}}\label{app:proof_th1}

If we write out the flows $\mathbf{S_\sigma q}$, we find
\begin{equation}\label{eq:SSigmaq}
    \mathbf{S_\sigma q} = \frac{1}{\sum x_{\ell}}
    \begin{bmatrix}
        -(x_N+x_E+x_S) & -(x_E+x_S) & -x_S & 0\\
        x_W & x_N+x_W & -x_S & 0\\
        x_W & x_N+x_W & -x_S & 0\\
        x_N+x_E+x_S & x_E+x_S & x_S & 0
    \end{bmatrix}
    \begin{bmatrix}
        q_N \\ -q_N \\ q_S \\ -q_S
    \end{bmatrix} = \frac{1}{\sum x_{\ell}}
    \begin{bmatrix}
        -x_Nq_N-x_Sq_S \\ -x_Nq_N-x_Sq_S \\ -x_Nq_N-x_Sq_S \\ x_Nq_N+x_Sq_S
    \end{bmatrix}.
\end{equation}
We can readily confirm that all flows are equal in magnitude. Next, we check the directions of these flows. If $(x_Nq_N+x_Sq_S) < 0$, the first three flows will be in the positive direction, and the last one will be in the negative direction, giving rise to a clockwise flow. When $(x_Nq_N+x_Sq_S) > 0$, the opposite happens and we have a counterclockwise flow being generated.

For part (c), we note that the line $\ell=N$ will have a tighter bound if the $\sigma$-flow $f^{\sigma }_N$ through that line is positive when the goods trade flow $f^g_{N}$ is also positive, or when they are both negative. For the positive case, we have the conditions
\begin{equation}
    -x_Nq_N-x_Sq_S > 0 ~ \text{and} ~ q_N > 0.
\end{equation}
If node $SE$ is also a supplier ($q_S>0$), the first condition cannot be satisfied. Then, $q_S$ must be negative and
\begin{equation}\label{eq:q1bound}
    0 < q_N < \frac{x_S}{x_N}|q_S|
\end{equation}
For the negative flow case, still for $\ell=N$, the conditions are
\begin{equation}
    -x_Nq_N-x_Sq_S < 0 ~ \text{and} ~ q_N < 0.
\end{equation}
Here, if both $q_N$ and $q_S$ are negative, the first condition is violated. Then, we must have $q_S>0$ and
\begin{equation}\label{eq:q3bound}
    0 < |q_N| < \frac{x_S}{x_N}q_S.
\end{equation}
We note \eqref{eq:q1bound} and \eqref{eq:q3bound} are equivalent and can be rewritten as the inequality presented in Theorem~\ref{th:sigmaf}(c). The proof for the transmission line  $\ell=S$ follows the same procedure as that presented for line $\ell=N$.

\section{Proof of Theorem~\ref{th:feasNS}}\label{app:proof_th2}
Due to the symmetry of the individual feasibility sets, we may turn our attention only to the intersection points in the left side of Table~\ref{tb:intersectionsNS}; if one of them is found to be a vertex, so is its symmetric intersection point (right column). Let $d_{n}$ denote the distance of the intersection $n$ to the origin. We can find that:
\begin{align}
    d_{1}^2 &= 2C_N^2-4C_NC_V\frac{x_V}{x_H}+C_V^2\left[1+\frac{(x_H+2x_V)^2}{x_H^2}\right]\\
    d_{3}^2 &= 2C_N^2+4C_NC_V\frac{x_V}{x_H}+C_V^2\left[1+\frac{(x_H+2x_V)^2}{x_H^2}\right]\\
    d_{5}^2 &= \frac{(C_N+C_S)^2}{2}\left(\frac{x_H+x_V}{x_V}\right)^2+\frac{(C_N-C_S)^2}{2}\\
    d_{7}^2 &= \frac{(C_N-C_S)^2}{2}\left(\frac{x_H+x_V}{x_V}\right)^2+\frac{(C_N+C_S)^2}{2}\\
    d_{9}^2 &= 2C_S^2-4C_SC_V\frac{x_V}{x_H}+C_V^2\left[1+\frac{(x_H+2x_V)^2}{x_H^2}\right]\\
    d_{11}^2 &= 2C_S^2+4C_SC_V\frac{x_V}{x_H}+C_V^2\left[1+\frac{(x_H+2x_V)^2}{x_H^2}\right]
\end{align}
The distances are such that, for any choice of network parameters, $d_{3}>d_{1}$, $d_{11}>d_{9}$, and $d_{5}>d_{7}$. To determine the feasible set $\mathcal{F}_{P} = \mathcal{F}_{P1} \cap \mathcal{F}_{PV} \cap \mathcal{F}_{P3}$, we will evaluate three distinct cases:
\begin{enumerate}
    \item Case $\mathcal{F}_{P1} \cap \mathcal{F}_{P3} \subset \mathcal{F}_{PV}$:

    The intersection $\mathcal{F}_{P1} \cap \mathcal{F}_{P3}$ is defined by the vertices $V_{5}$, $V_{6}$, $V_{7}$ and $V_{8}$. For the polygon that limits $\mathcal{F}_{P}$ to be a quadrilateral, the boundary lines 3 and 4 must not intersect $\mathcal{F}_{P1} \cap \mathcal{F}_{P3}$ (i.e. $\mathcal{F}_{P1} \cap \mathcal{F}_{P3} \subset \mathcal{F}_{PV}$). To determine that condition, it suffices to find when line 4 is outside the circle centered in the origin and with radius $\max(d_{5},d_{7})=d_{5}$ (by symmetry, the same will be true for line 3). Let
\begin{equation}
    d_{4o} = \frac{C_V}{X_r \sqrt{2}}=\frac{2C_V(x_H+x_V)}{x_H\sqrt{2}}
\end{equation}
be the distance from line 4 to the origin. Then, $d_{4o} > d_{5} \implies d_{4o}^2 > d_{5}^2$, which leads to
\begin{equation}
\frac{2C_V^2(x_H+x_V)^2}{x_H^2} > \frac{(C_N+C_S)^2}{2}\left(\frac{x_H+x_V}{x_V}\right)^2+\frac{(C_N-C_S)^2}{2},    
\end{equation}
from which we can find the inequality condition in \eqref{eq:cxratioNS}. Furthermore, it can be shown that, if we have an equality in the above condition, then the intersections $1$, $5$ and $9$ will coincide, which means we still have a quadrilateral. Therefore, if \eqref{eq:cxratioNS} holds, the quadrilateral with vertices $V_{5}$, $V_{6}$, $V_{7}$ and $V_{8}$ limits the feasibility region for power trades.

\item Case $\mathcal{F}_{P1} \cap \mathcal{F}_{PV} \subset \mathcal{F}_{P3}$:

The intersection $\mathcal{F}_{P1} \cap \mathcal{F}_{PV}$ is defined by the vertices $V_{1}$, $V_{2}$, $V_{3}$ and $V_{4}$. This case will happen when lines 5 and 6 do not intersect $\mathcal{F}_{P1} \cap \mathcal{F}_{PV}$. Similarly to the previous case, the condition is
\begin{equation}
   d_{6o} \geq \max(d_{1},d_{3})= d_{3} \implies d_{6o}^2 \geq d_{3}^2,
\end{equation}
where $d_{6o}$ is the distance from line 6 to the origin. Then, if
\begin{equation}
   \frac{2C_S^2(x_H+x_V)^2}{x_H^2+2x_V^2} \geq 2C_N^2+4C_NC_V\frac{x_V}{x_H}+C_V^2\left[1+\frac{(x_H+2x_V)^2}{x_H^2}\right],
\end{equation}
the polygon that limits the set $\mathcal{F}_{P}$ is the quadrilateral $\mathcal{F}_{P1} \cap \mathcal{F}_{PV}$.

\item Case $\mathcal{F}_{PV} \cap \mathcal{F}_{P3} \subset \mathcal{F}_{P1}$:

In this case, $\mathcal{F}_{PV} \cap \mathcal{F}_{P3}$ is limited by the quadrilateral with vertices $V_{9}$, $V_{10}$, $V_{11}$ and $V_{12}$. Further, the feasibility set for the entire network will be $\mathcal{F}_{P} = \mathcal{F}_{PV} \cap \mathcal{F}_{P3}$ if
\begin{equation}
   d_{2o} \geq \max(d_{9},d_{11})= d_{11} \implies d_2^{2o} \geq d_{11}^2,
\end{equation}
where $d_{2o}$ is the distance from line 2 to the origin. That condition is equivalent to
\begin{equation}
   \frac{2C_N^2(x_H+x_V)^2}{x_H^2+2x_V^2} \geq 2C_S^2+4C_SC_V\frac{x_V}{x_H}+C_V^2\left[1+\frac{(x_H+2x_V)^2}{x_H^2}\right].
\end{equation}
\end{enumerate}

If none of the above conditions hold, then the set $\mathcal{F}_{P}$ will be a hexagon (Case 4) delimited by the vertices $V_{1}$, $V_{2}$, $V_{7}$, $V_{8}$ , $V_{9}$ and $V_{10}$, which are the intersection points of each two lines that are closer to the origin and their symmetric points.


\section{Feasible Region: Symmetric network}\label{sec:symmetric}
Consider the case of a symmetric network in which the capacities and reactances of the horizontal lines are the same (i.e. $C_N=C_S=C_H$ and $x_N=x_S=x_H$), and the same is true for the vertical lines (i.e. $C_E=C_W=C_V$ and $x_E=x_W=x_V$). Note that this is a special case of the more general network presented in Section~\ref{sec:regionfeasibility}, and therefore we determine the feasible region in the same way.

It is straightforward to notice from \eqref{eq:f23} that the sets of feasible trades for the vertical lines $E$ and $W$ will be the same. Henceforth, we denote it $\mathcal{F}_{PV}$. The coordinates for the intersections $a-b$ between boundary lines $a,b\in\mathcal{B}$, which are equivalent to the intersections $b-a$, are listed in Table~\ref{tb:intersections}, where the boundaries are numbered as in Section~\ref{sec:regionfeasibility}. These intersection points can be found by setting $C_N=C_S=C_H$ in the values presented in Table~\ref{tb:intersectionsNS}. The results in Theorem~\ref{th:feasNS} can be used to find the following result for the symmetric network.

\begin{table}[h]
    \centering
    \setlength{\extrarowheight}{10pt}
\begin{tabular}{ c  c  c  c}
        $n$ & $(q_{N},q_{S})$ & $n$ & $(q_{N},q_{S})$ \\ \hline
        1 & $\left(C_H+C_V,\frac{(1-X_r)}{X_r}C_V-C_H\right)$ & 2 & $\left(-C_H-C_V,-\frac{(1-X_r)}{X_r}C_V+C_H\right)$  \\ 
        3 & $\left(C_H-C_V,-\frac{(1-X_r)}{X_r}C_V-C_H\right)$  & 4 & $\left(C_V-C_H,\frac{(1-X_r)}{X_r}C_V+C_H\right)$ \\ 
        5 & $\left(\frac{C_H}{1-2X_r},\frac{C_H}{1-2X_r}\right)$ & 6 & $\left(\frac{-C_H}{1-2X_r},\frac{-C_H}{1-2X_r}\right)$ \\ 
        7 & $\left(C_H,-C_H\right)$ & 8 & $\left(-C_H,C_H\right)$ \\ 
        9 & $\left(\frac{(1-X_r)}{X_r}C_V-C_H,C_H+C_V\right)$ & 10 & $\left(-\frac{(1-X_r)}{X_r}C_V+2C_H,-C_H-C_V\right)$  \\ 
        11 & $\left(\frac{(1-X_r)}{X_r}C_V+C_H,C_V-C_H\right)$  & 12 & $\left(-\frac{(1-X_r)}{X_r}C_V-C_H,C_H-C_V\right)$  \\
      \hline
    \end{tabular}
\caption{Intersection points between boundaries of feasible regions $\mathcal{F}_{P1}$, $\mathcal{F}_{PV}$, and $\mathcal{F}_{P3}$.}
\label{tb:intersections}
\end{table}

\begin{corollary}\label{cor:feasS}
For a symmetric 4-node electricity grid with $C_N=C_S=C_H$, $x_N=x_S=x_H$, $C_E=C_W=C_V$, and $x_E=x_W=x_V$, if
\begin{equation}\label{eq:cxratio}
    \frac{C_V}{C_H}\geq \frac{x_H}{x_V},
\end{equation}
then the set of feasible power trades in the network is the quadrilateral defined by the vertices $V_{5}$, $V_{6}$, $V_{7}$ and $V_{8}$, where $V_{n}$ is the intersection point $n$, as presented in Table~\ref{tb:intersections}. If \eqref{eq:cxratio} does not hold, the feasible region is the hexagon defined by the vertices $V_{1}$, $V_{2}$, $V_{7}$, $V_{8}$ , $V_{9}$ and $V_{10}$.
\end{corollary}
\noindent Proof is in Appendix~\ref{app:proof_cor2}.

Note that the intersections $7$ and $8$ lie on the line $q_{N}=-q_{S}$, and thus this realization will give rise to trades that are equivalent to those in a goods network, as stated in Corollary~\ref{col:floweq}. That means the flows in the vertical lines will be zero and never violate the capacity constraints of these lines. Thus, these intersections will always be within the feasible region $\mathcal{F}_{PV}$. Further, in this scenario, the power flows that arise in lines $\ell=N$ and $\ell=S$ are equal to the maximum capacity allowed in these lines; therefore, intersections $7$ and $8$ will always be on the boundaries of the feasible regions $\mathcal{F}_{P1}$ and $\mathcal{F}_{P3}$ simultaneously. It follows that these intersection points will always be vertices of $\mathcal{F}_{P}$. On the other hand, intersections $5$ and $6$ may removed from the set of vertices of the polygon when the boundary lines 3 and 4 intersect $\mathcal{F}_{P1} \cap \mathcal{F}_{P3}$, as illustrated in Figure~\ref{fig:FEset}.

\begin{figure*}[h]
\begin{multicols}{3}
    \includegraphics[width=\linewidth]{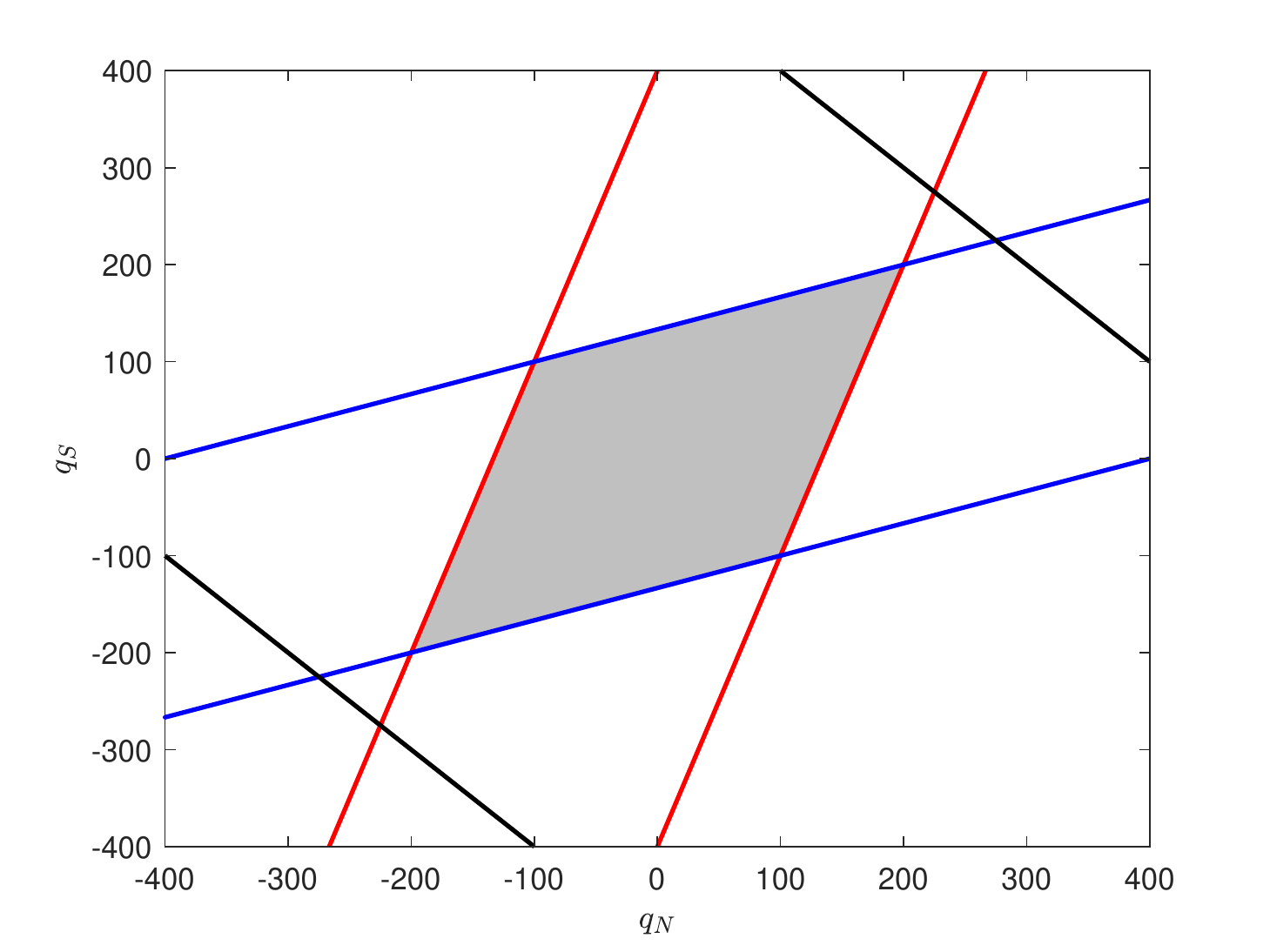}\par
    \includegraphics[width=\linewidth]{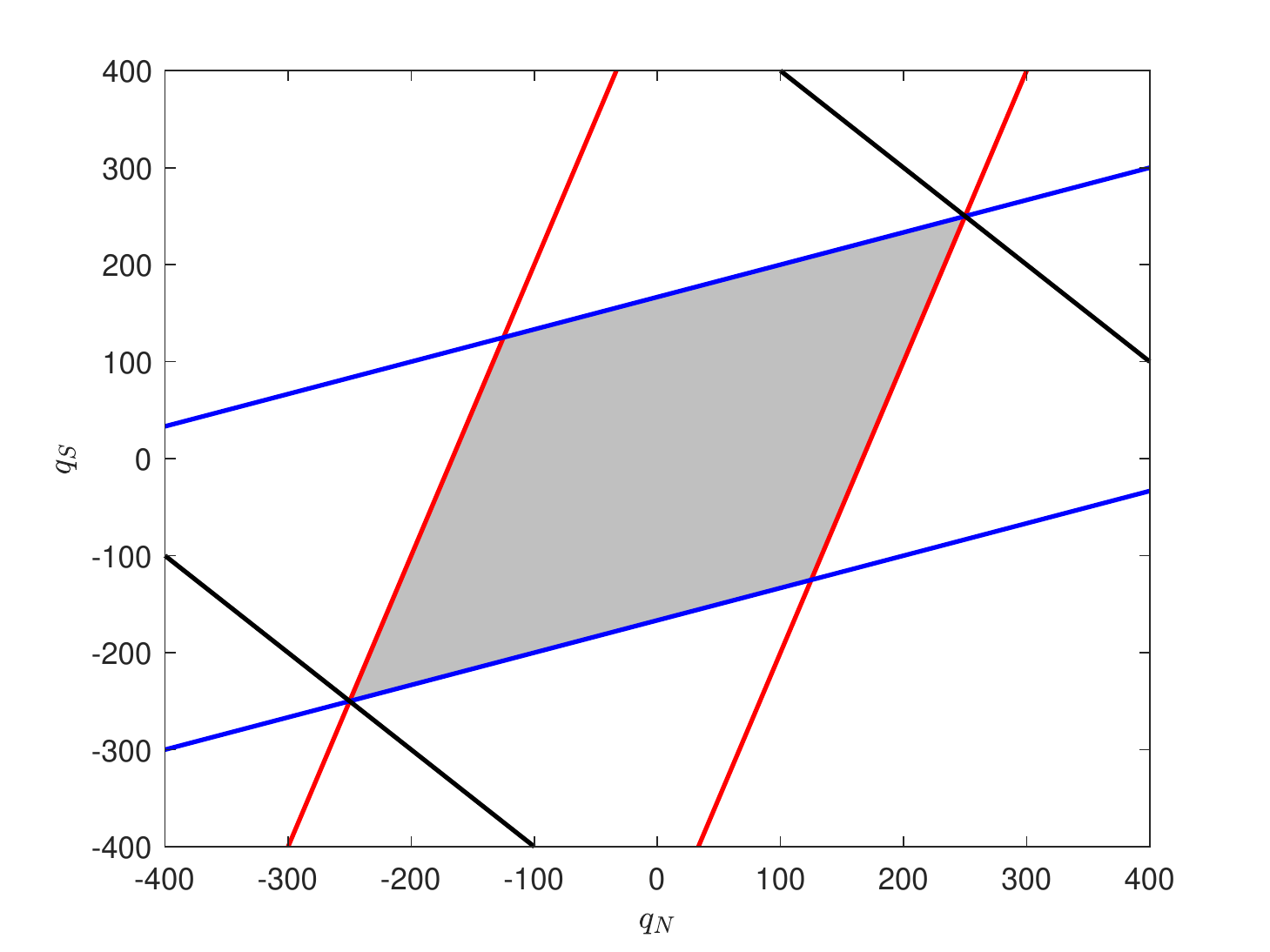}\par
    \includegraphics[width=\linewidth]{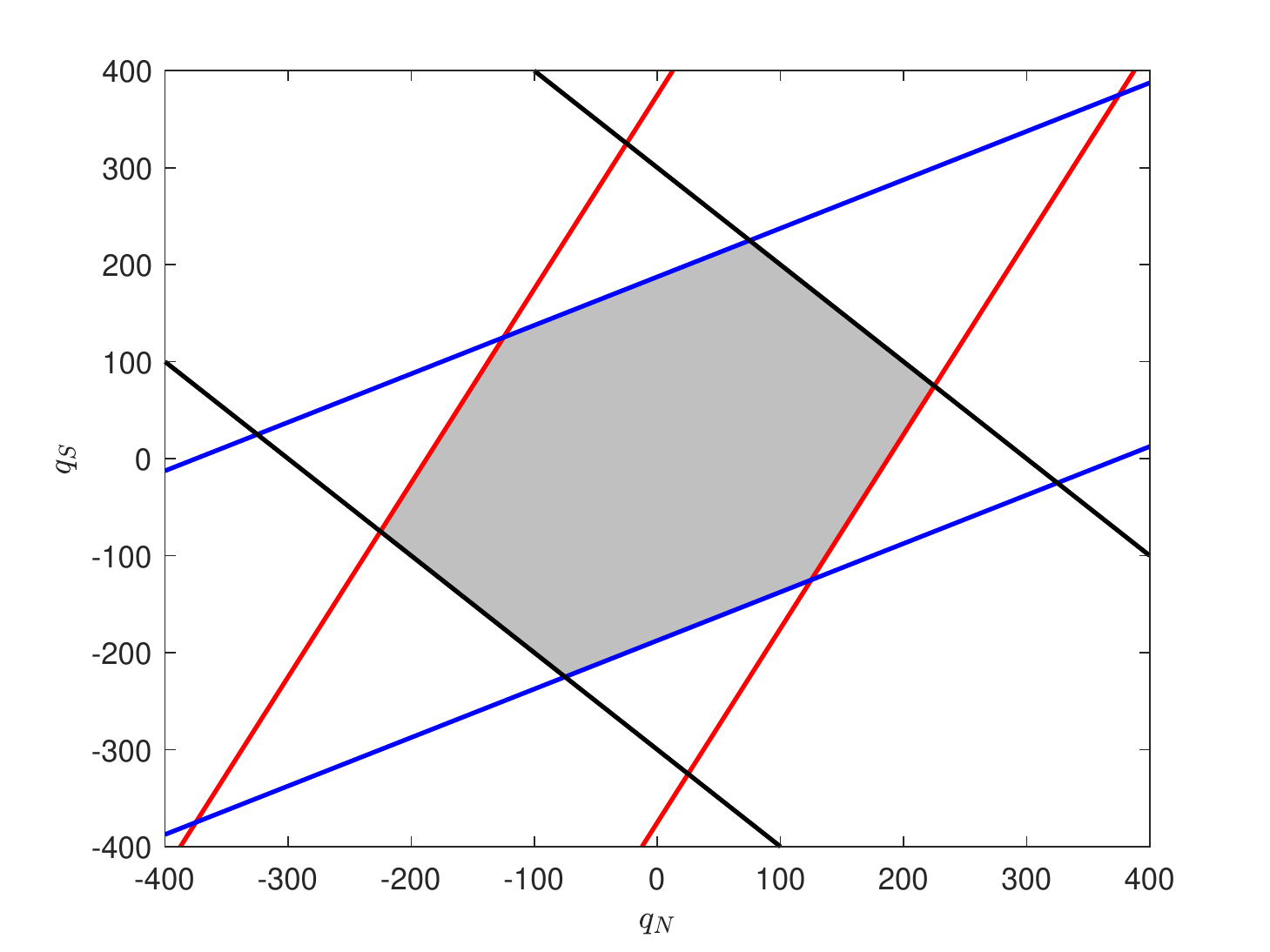}\par
\end{multicols}
\caption{Feasibility sets for power trades in a symmetric 4-node network. The red, black and blue lines limit the sets $\mathcal{F}_{P1}$, $\mathcal{F}_{PV}$, and $\mathcal{F}_{P3}$, respectively. On the left, we have the case with $\frac{C_V}{C_H} > \frac{x_H}{x_V}$, in which lines 3 and 4 do not intersect $\mathcal{F}_{P1} \cap \mathcal{F}_{P3}$. In the center, we have the limiting case in which $\frac{C_V}{C_H} = \frac{x_H}{x_V}$, and we can see the three intersection points that coincide to become a vertex of the quadrilateral. The last figure shows the hexagon case, with $\frac{C_V}{C_H} < \frac{x_H}{x_V}$.}
\label{fig:FEset}
\end{figure*}

In the symmetric case, the feasible set for goods trades $\mathcal{F}_G$ is a square, which can be found by the definition of this set in \eqref{eq:fp}. In this scenario, it can be observed that $\mathcal{F}_G$ will always be inside the feasible set for the power trades for a large enough $C_V$. This can be noted because, for the power trades, the boundaries corresponding to transmission lines $\ell=N$ and $\ell=S$ will both rotate around the points $(-C_H,C_H)$ and $(C_H,-C_H)$, as described in Appendix~\ref{app:rotation}. These points are the corners of $\mathcal{F}_G$ that are located in the second and fourth quadrants, and they are also two of the vertices of the parallelogram defined by $\mathcal{F}_P$. The other two vertices of the feasible set for power trades, located in the first ($V_{5}$) and third ($V_{6}$) quadrants, will always be above and to the right (if in the first quadrant) and below and to the left (if in the third quadrant) of the vertices of the square goods feasible set.

\section{Proof of Corollary~\ref{cor:feasS}}\label{app:proof_cor2}

We proceed similarly to the proof of Theorem~\ref{th:feasNS}. The distances from the intersections to the origin are
\begin{align}
    d_{1}^2 &= d_{9}^2 = 2C_H^2-4C_HC_V\frac{x_V}{s_H}+C_V^2\left[1+\frac{(x_H+2x_V)^2}{x_H^2}\right]\\
    d_{3}^2 &= d_{11}^2 = 2C_H^2+4C_HC_V\frac{x_V}{s_H}+C_V^2\left[1+\frac{(x_H+2x_V)^2}{x_H^2}\right]\\
    d_{5}^2 &= 2\frac{(x_H+x_V)^2}{x_V^2}C_H^2\\
    d_{7}^2 &= 2C_H^2
\end{align}
The distances are such that, for any choice of network parameters, $d_{3}=d_{11}>d_{1}=d_{9}$ and $d_{5}>d_{7}$.

To determine the feasible set $\mathcal{F}_{P} = \mathcal{F}_{P1} \cap \mathcal{F}_{PV} \cap \mathcal{F}_{P3}$, let us first take the intersection $\mathcal{F}_{P1} \cap \mathcal{F}_{P3}$, which is defined by the vertices $V_{5}$, $V_{6}$, $V_{7}$ and $V_{8}$. For the polygon that limits $\mathcal{F}_{P}$ to be a quadrilateral, boundary lines 3 and 4 must not intersect $\mathcal{F}_{P1} \cap \mathcal{F}_{P3}$ (i.e. $\mathcal{F}_{P1} \cap \mathcal{F}_{P3} \subset \mathcal{F}_{PV}$). To determine that condition, it suffices to find when line 4 is outside the circle centered in the origin and with radius $\max(d_{5},d_{7})=d_{5}$ (by symmetry, the same will be true for line 3). Let
\begin{equation}
    d_{4o} = \frac{C_V}{X_r \sqrt{2}}=\frac{2C_V(x_H+x_V)}{x_H\sqrt{2}}
\end{equation}
be the distance from line 4 to the origin. Then,
\begin{equation}
   d_{4o} > d_{5} \implies \frac{2C_V(x_H+x_V)}{x_H\sqrt{2}} > \sqrt{2}\frac{(x_H+x_V)}{x_V}C_H \implies \frac{C_V}{C_H} > \frac{x_H}{x_V}.
\end{equation}
Furthermore, it can be shown that, if we have an equality in the above condition, then the intersections $1$, $5$ and $9$ will coincide, which means we still have a quadrilateral. Therefore, if \eqref{eq:cxratio} holds, the quadrilateral with vertices $V_{5}$, $V_{6}$, $V_{7}$ and $V_{8}$ limits the feasibility region for power trades.

If that condition does not hold, lines 3 and 4 will intersect $\mathcal{F}_{P1} \cap \mathcal{F}_{P3}$, and these boundaries become two new faces that limit the set $\mathcal{F}_{P}$. In that scenario, $d_{4o} < d_{5}$, and thus $5$ (and its symmetrical $6$) will no longer be vertices of the polygon. Because $d_{3}=d_{11}>d_{1}=d_{9}$, the four new vertices in this case will be $V_{1}$ and $V_{9}$, along with the symmetrical points $V_{2}$ and $V_{10}$, which completes the proof.

\section{Non-Symmetric Network: From Goods to Power (and vice-versa)}\label{app:rotation}
From Theorem~\ref{th:feasNS}, we can notice the feasible set for the electricity network will fall into Case 1 for a large enough $C_V$. In the next result, we focus on that special case.

\begin{corollary}\label{cor:rotation}
For a 4-node electricity grid with $x_N=x_S=x_H$, $x_E=x_W=x_V$, and large enough capacity $C_V$, the parallelogram that limits $\mathcal{F}_{P}$ is formed by rotating the individual feasible sets for the goods trades. For slopes as in \eqref{eq:m1}, the parallel boundaries limiting the power trades in:
\begin{enumerate}
    \item transmission line $\ell=N$ are a clockwise rotation of $\alpha_{N}=\arctan(1/m_1)$ of the goods trades for the same line, around the point $(C_N,-C_N)$ for the boundary with positive binding flow and the point $(-C_N,C_N)$ for the one with negative binding flow;
    \item transmission line $\ell=S$ are a counterclockwise rotation of $\alpha_{S}=\arctan(m_3)$ of the goods trades for the same line, around the point $(-C_S,C_S)$ for the boundary with positive binding flow and the point $(C_S,-C_S)$ for the one with negative binding flow.
\end{enumerate}
Furthermore, because of the symmetry of the reactances, the rotation angles are the same ($\alpha_{N}=\alpha_{S}$) and $\mathcal{F}_{P} \to \mathcal{F}_{G}$ as $x_V/x_H \to \infty$.
\end{corollary}

\begin{proof}
From \eqref{eq:m1}, we observe that the slopes are
\begin{equation*}
    m_1=1+2x_V/x_H ~~ \text{and} ~~ m_3=\frac{1}{1+2x_V/x_H}=\frac{1}{m_1},
\end{equation*}
which are always positive and dependent on the reactances only. The rotation angle for the first boundary lines are taken with respect to the y-axis, and thus can be written as $\alpha_{N}=90\degree-\arctan(m_1)=\arctan(1/m_1)$, since the slope is positive. For the second pair of boundary lines, the reference is the x-axis, and thus we have $\alpha_{S}=\arctan(m_3)$. Since $m_3=1/m_1$, we can note that those angles are the same. Further, as $x_V/x_H \to \infty$, $m_1 \to \infty$ and those boundary lines tend to become vertical and coincide with the boundaries of the goods trades for that transmission line. Similarly, $m_3\to 0$ for the same condition, and those boundary lines tend to coincide with the horizontal boundaries of the goods trades for line $\ell=S$.

To find the fixed points in the rotations, we begin by finding all the intersection points between the boundary lines that form the goods square set and those that form the power parallelogram set (for brevity, they are not listed here). For each power boundary, there is only one intersection which is invariant with reactance changes, and thus with rotation. Therefore, only those intersections, which are presented in Corollary~\ref{cor:rotation}, can be the fixed points.
\end{proof}

The results in Corollary~\ref{cor:rotation} are illustrated in Figure~\ref{fig:rotation}, where the green lines correspond to the boundaries of the feasible region for the goods trades and the grey lines intersecting the origin mark the points around which the boundary lines rotate.

\begin{figure*}[tpb]
\begin{multicols}{2}
    \centering
    \includegraphics[width=0.75\linewidth]{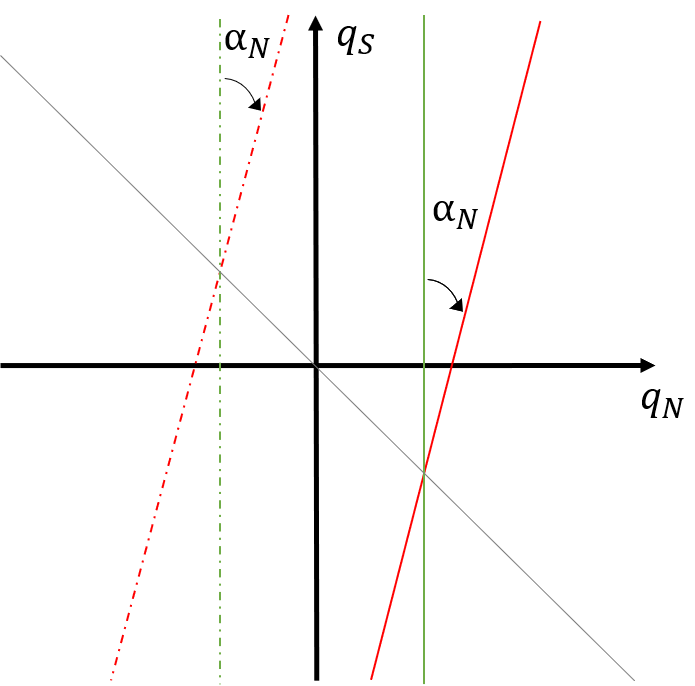}\par
    \includegraphics[width=0.75\linewidth]{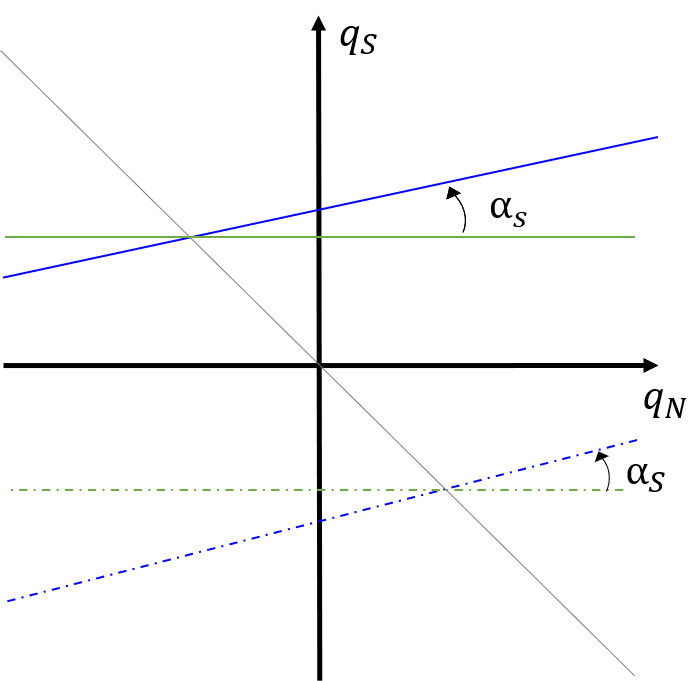}\par
\end{multicols}
\caption{Illustration of Corollary~\ref{cor:rotation} showing how the boundaries of the set of feasible power trades rotate with respect to the boundaries of the set of feasible goods trades.}
\label{fig:rotation}
\end{figure*}

\begin{remark}
The intersection points that are fixed with the rotation correspond to trades that generate equivalent flows in the electricity and the goods trade network. From Corollary~\ref{col:floweq}, we observe that those are the trades with $q_{N}=-q_{S}$ for this symmetric network.
\end{remark}

\begin{remark}
The intersections of the parallelogram with the square are such that, as the boundary lines for the power trades rotate, trades in the second and fourth quadrants that would be feasible in a goods network are not feasible in the power network. As illustrated in Figure~\ref{fig:Sflows}, those are the quadrants in which the flows in the lines $\ell=N$ and $\ell=S$ are tighter. Further, with these rotations, some trades that are not feasible in the goods case are added in the set of feasible power trades.
\end{remark}

The observations from the remarks above can be visualized in Figure~\ref{fig:intNS}. The red lines correspond to feasible sets for transmission line $\ell=N$, while the blue ones refer to line $\ell=S$. The dashed lines are the boundaries for the goods trades and the (red and blue) solid lines represent the feasible sets for power trades. The solid black like shows the trades that produce the same flows both in the power and in the goods networks. Note that the black line intersects the fixed points around which the boundary lines rotate. Furthermore, we can observe that $\mathcal{F}_{P}$ contains some trades that are not allowed in $\mathcal{F}_{G}$, and vice-versa.

\begin{figure*}[h]
    \centering
    \includegraphics[width=0.5\textwidth]{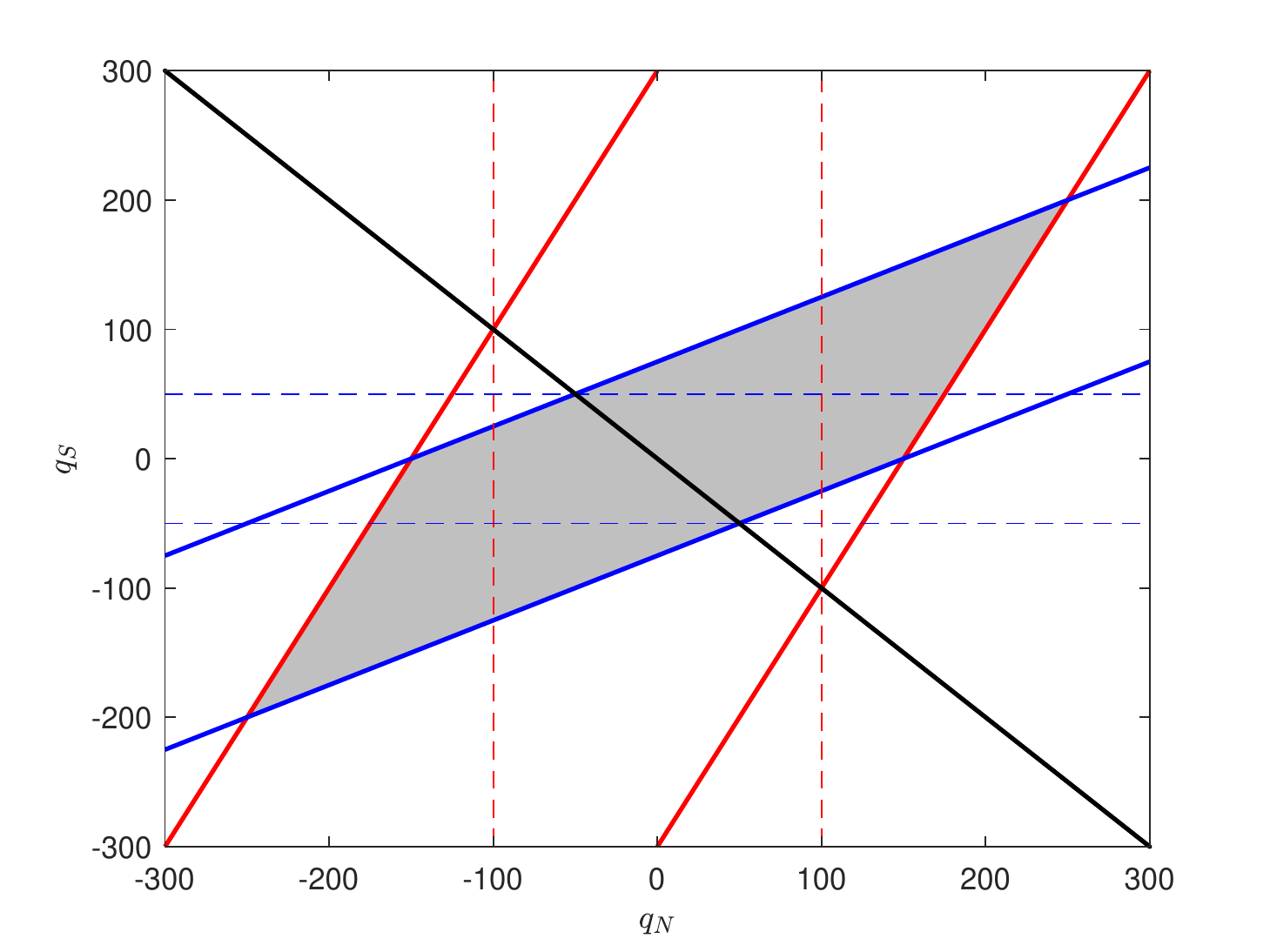}
\caption{Example of feasible sets for power (shaded) and goods trades in a 4-node network, for large enough $C_V$.}
\label{fig:intNS}
\end{figure*}

\section{Further Decomposition Of Network Externalities}
\label{sec:further}

Since we adopt a linear power flow model, we can use the superposition principle to verify how each individual trade contributes to the externality. We first decompose the power flows that arise due to the 4735 MW North-North trade to find that it produces a clockwise extra flow of magnitude 1274.8 MW. 

Similarly, we can decompose the power flows of the 2931 MW South-South trade, which generates a 789.1 MW counterclockwise extra flow that reduces the available capacity in the North line. Therefore, each individual trade is causing a negative externality to the transmission line that connects the remaining two nodes. The power flow decomposition for both trades is as below.
\begin{align*}
    \text{North-North trade:} ~\bm{f^p} &= \bm{f^g} + \mathbf{S_\sigma}~\mathbf{q} = \begin{bmatrix}
    -4735 \\ 0 \\ 0 \\ 0
    \end{bmatrix}+
    \begin{bmatrix}
    1274.8 \\ 1274.8 \\ 1274.8 \\ -1274.8\\
    \end{bmatrix} = \begin{bmatrix}
    -3460.2 \\ 1274.8 \\ 1274.8 \\ -1274.8\\
    \end{bmatrix}\\
    \text{South-South trade:} ~\bm{f^p} &= \bm{f^g} + \mathbf{S_\sigma}~\mathbf{q} = \begin{bmatrix}
    0 \\ 0 \\ 2931 \\ 0
    \end{bmatrix}+
    \begin{bmatrix}
    -789.1 \\ -789.1 \\ -789.1 \\ 789.1\\
    \end{bmatrix} = \begin{bmatrix}
    -789.1 \\ -789.1 \\ 2141.9 \\ 789.1\\
    \end{bmatrix}
\end{align*}
Note that the sum of the power flows generated from both trades is indeed the total power flow calculated previously, as per the superposition principle. From this individual analysis, we observe that the North-North trade imposes a more significant externality than the South-South trade. As a consequence, as described previously, the overall extra flow is clockwise and has a magnitude of 485.7MW (1274.8MW - 789.1MW). 

\section{Renewable Energy Production in Europe over time}

\begin{figure*}[h]
    \centering
    \includegraphics[width=0.8\textwidth]{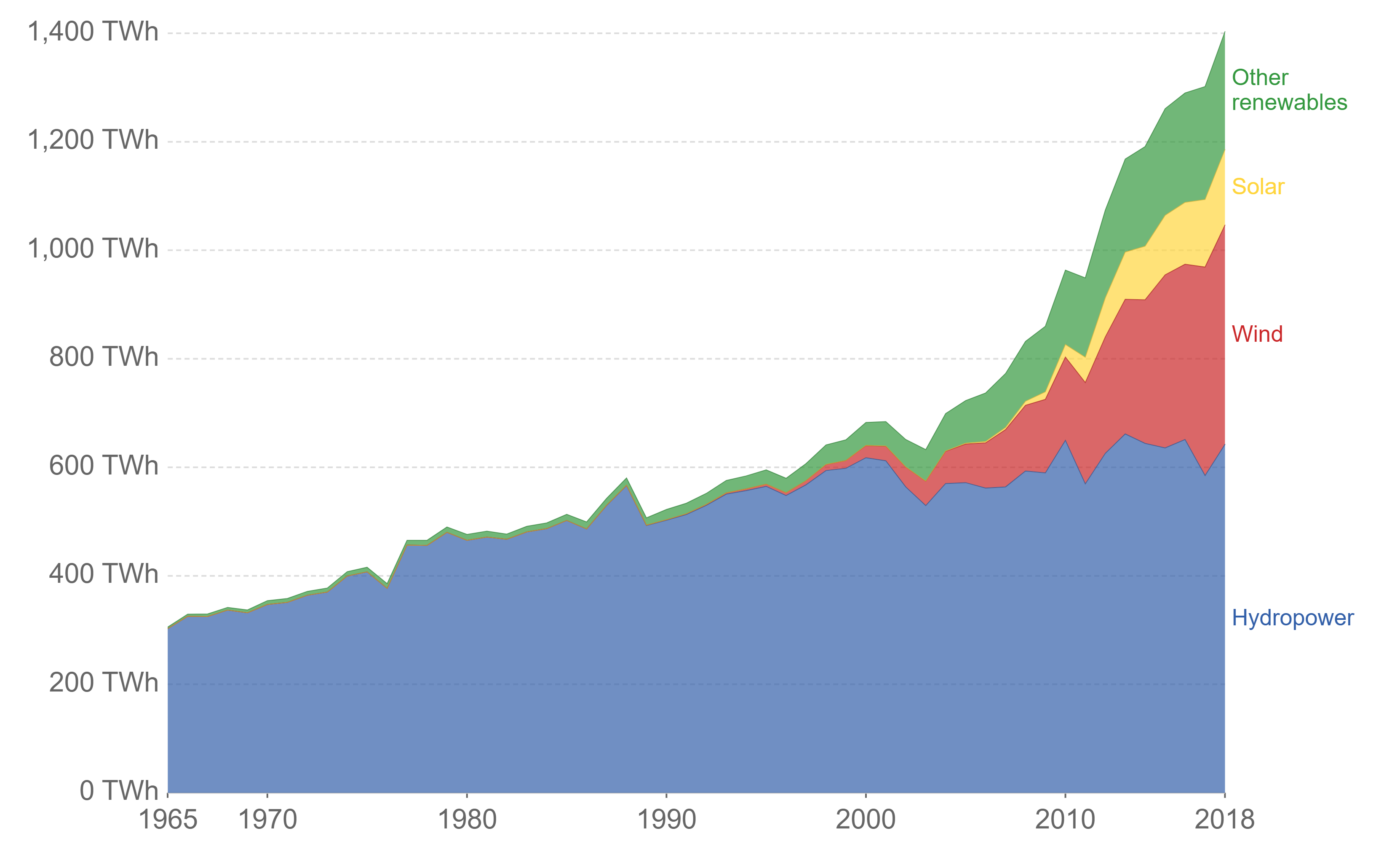}
\caption{Renewable energy generation, Europe, 1965 to 2018}
\quote{Source: BP Statistical Review of Global Energy (2019) and OurWorldInData.org/renewable-energy. Note: `Other renewables' refers to renewable sources including geothermal, biomass, waste, wave, and tidal energy.}
\label{fig:renewabletrends}
\end{figure*}

\end{document}